\documentclass[10pt]{article}
\usepackage{geometry}
\usepackage{authblk} 

\usepackage{microtype}
\DeclareMathAlphabet{\mathbbold}{U}{bbold}{m}{n}

\usepackage[usenames,dvipsnames,svgnames]{xcolor}
\definecolor{light-gray}{gray}{0.90}
\definecolor{lightblue}{rgb}{.90,.95,1}
\definecolor{lightred}{rgb}{1,.5,0}
\definecolor{mygreen}{rgb}{0,0.6,0}
\definecolor{mygray}{rgb}{0.5,0.5,0.5}
\definecolor{mymauve}{rgb}{0.58,0,0.82}
\definecolor{armygreen}{rgb}{0.29, 0.33, 0.13}
\definecolor{aqua}{rgb}{0.0, 1.0, 1.0}
\definecolor{blue-violet}{rgb}{0.54, 0.17, 0.89} 	 	
\definecolor{cerulean}{rgb}{0.0, 0.48, 0.65} 
\definecolor{coolblack}{rgb}{0.0, 0.18, 0.39}
\definecolor{coquelicot}{rgb}{1.0, 0.22, 0.0}
\definecolor{coral}{rgb}{1.0, 0.5, 0.31}
\definecolor{darkgoldenrod}{rgb}{0.72, 0.53, 0.04}

\usepackage{amsmath,amssymb,amsfonts,latexsym,amsthm}
\theoremstyle{plain}
\newtheorem{proposition}{Proposition}
\newtheorem{fact}{Fact}
\newtheorem*{notation*}{Notation and terminology}
\newtheorem{theorem}{Theorem}
\newtheorem{corollary}{Corollary}
\newtheorem{lemma}{Lemma}
\newtheorem{definition}{Definition}
\newtheorem{remark}{Remark}

\usepackage{stmaryrd}

\usepackage[goodsyntax]{virginialake}
\usepackage{hyperref}
\hypersetup{
    colorlinks,
    linkcolor={red!50!black},
    citecolor={blue!50!black},
    urlcolor={blue!80!black}
}

\usepackage{float}
\floatstyle{boxed} 
\restylefloat{figure}

\usepackage{graphicx}
\usepackage{wrapfig}

\usepackage{txfonts} 
\usepackage{subcaption} 
\usepackage{xspace}
\usepackage{rotating}
\usepackage{setspace}
\usepackage{rotating}
\usepackage{cmll} 
\usepackage{multirow}
\usepackage{tikz}
\usepackage{adjustbox}
\usepackage{dsfont} 
\usepackage{tabularx}

\usepackage{tikz}
\usetikzlibrary{decorations.markings}
\usetikzlibrary{shapes.geometric, matrix}
\usetikzlibrary{matrix}
\usetikzlibrary{calc}
\pgfdeclarelayer{edgelayer}
\pgfdeclarelayer{nodelayer}
\pgfsetlayers{edgelayer,nodelayer,main}

\tikzstyle{none}=[inner sep=1pt]

\tikzstyle{simple}=[-,draw=Black,line width=.500]
\tikzstyle{boxed} = [-,draw=Black, rectangle, text centered, minimum 
height=8mm,node distance=10em]
\tikzstyle{arrowthick}=[-,draw=Black,postaction={decorate},decoration={markings,mark=at
 position 1 with {\arrow{>}}},line width=1.500]
\tikzstyle{arrow}=[-,draw=Black,postaction={decorate},decoration={markings,mark=at
 position .5 with {\arrow{>}}},line width=.500]
\tikzstyle{tick}=[-,draw=Black,postaction={decorate},decoration={markings,mark=at
 position .5 with {\draw (0,-0.1) -- (0,0.1);}},line width=.500]

\renewcommand{\frac}[2]{ %
\vlinf{}{}
{#2}
{#1}}

\newcommand\APEX[5]{
\ifthenelse{\boolean{proofintheappendix}}

{\ifthenelse{\equal{#2}{}}{\begin{#1}}{\begin{#1}[#2]}\label{#3}
\leavevmode\marginpar{\textit{\small Proof at page
\pageref{APXZ#3}.}}{#4}\end{#1}\immediate\write\tempfile{
    \unexpanded{\subsubsection*{Proof of \uppercase{#1} \ref{#3}, page
\pageref{#3}}\label{APXZ#3}{#5}}}
    }{\ifthenelse{\boolean{bodyafterstatement}}

{\ifthenelse{\equal{#2}{}}{\begin{#1}}{\begin{#1}[#2]}\label{#3}{#4}\end{#1}{#5}}

{\ifthenelse{\equal{#2}{}}{\begin{#1}}{\begin{#1}[#2]}\label{#3}{#4}\end{#1}}
 }}

\newcommand{\Set}[1]{ \left\lbrace #1 \right\rbrace} 
\newcommand{\qedh}{\hfill $ \bullet $}

\newcommand{\ie}{i.e.\xspace}
\newcommand{\Ie}{I.e.\xspace}

\definecolor{Blue}{rgb}{0,0,1}

\vllineartrue

\renewcommand{\vlne}[1]{\overline{#1}}
\newcommand{\bTT}{\textcolor{blue}{\texttt{T}}}
\newcommand{\bFF}{\textcolor{red} {\texttt{F}}}

\newcommand{\atma}{a}

\newcommand{\natma}{\vlne\atma}

\begin{document}
\title{Subatomic systems need not be subatomic}
\author{Luca Roversi}
\date{%
\small
Dipartimento di Informatica -- Università di Torino\\
C.so Svizzera 185, 10149 Torino, ITALY\\
roversi@di.unito.it, lroversi@unito.it\\
\url{https://orcid.org/0000-0002-1871-6109}\\[2ex]
\normalsize
\today
}

\maketitle

\begin{abstract}
Subatomic systems were recently introduced to identify the structural  
principles underpinning the normalization of proofs. ``Subatomic'' means 
that we can reformulate logical systems in accordance with two 
principles. Their atomic formulas become instances of sub-atoms, i.e. of 
non-commutative self-dual relations among logical constants, and their rules 
are derivable by means of a unique deductive scheme, the medial shape. One of 
the neat results is that the cut-elimination of subatomic systems implies the 
cut-elimination of every standard system we can represent sub-atomically.

We here introduce Subatomic systems-1.1. They relax and widen the properties 
that the sub-atoms of Subatomic systems can satisfy while maintaining the use 
of the medial shape as their only inference principle. Since sub-atoms can 
operate directly on variables we introduce $ \mathsf{P} $. The cut-elimination 
of $ \mathsf{P} $ is a corollary of the cut-elimination that we prove for  
Subatomic systems-1.1.  Moreover, $ \mathsf{P} $ is sound and complete with 
respect to the clone at  the top of Post's Lattice. I.e. $ \mathsf{P} $ proves 
all and only the tautologies that contain conjunctions, disjunctions and 
projections. So, $ \mathsf{P} $ extends Propositional logic without any 
encoding of its atoms as sub-atoms of $ \mathsf{P} $. 

This shows that the logical principles underpinning Subatomic systems also 
apply outside the sub-atomic level which they are conceived to work at.
We reinforce this point of view by introducing the set of medial shapes 
$\mathsf{R_{23}} $. The formulas that the rules in $\mathsf{R_{23}} $ deal with 
belong to the union of two disjoint clones of Post's Lattice.  The SAT-problem 
of the first clone is in P-Time. The SAT-problem of the other is NP-Time 
complete. So, $ \mathsf{R_{23}}$ and the proof technology of Subatomic systems 
could help to identify proof-theoretical properties that highlight the phase 
transition from P-Time to NP-Time complete satisfiability.
\end{abstract}

\section{Introduction}
\label{section;Introduction}
\vllinearfalse
Subatomic systems were recently introduced to identify the structural  
principles underpinning the normalization of proofs.  
``Subatomic'' means that we can reformulate logical systems in accordance with 
two principles. The atomic constituents of the formulas become instances of 
sub-atoms, i.e. of non-commutative self-dual relations among logical constants, 
and the rules are derivable by means of a unique deductive scheme, the medial 
shape.
In its not full, but general enough, form it is:
    \begin{align*}
    &&
     \vlinf{}{}
     {(A\vlbin{\beta}C)\vlbin{\alpha}(B\vlbin{\delta}D)}
     {(A\vlbin{\alpha}B)\vlbin{\beta}(C\vlbin{\gamma}D)}
    \end{align*}
where $ A, B, C, D $ are formulas and $ \alpha, \beta, \gamma, \delta $ 
relations. 
For example, let us focus on propositional logic. The sub-atomic rule
$ 
{\vlinf{}{}
{\vls(A\vlan C)\vlbin{\vlor}(B\vlbin{\vlor}D)}
{\vls(A\vlbin{\vlor}B)\vlbin{\vlan}(C\vlbin{\vlor}D)}} $
stands for the introduction to the right of the conjunction. It is a rule in 
deep inference which we can read as follows. Let 
$ A\vlbin{\vlor}B $ and
$ C\vlbin{\vlor}D $ be two given disjunctions where $ B $ is the premise that 
allows to derive $ A $ and $ D $ the one for deriving $ C  $. Then, the rule 
derives $ A\vlbin{\vlan}C $ from the premise
$ B\vlbin{\vlor}D $. The sub-atomic 
rule 
$ \vlinf{}{}
{\vls(\mathtt{f}\vlbin{\mathbf{a}}\mathtt{t})
    \vlbin{\vlor}
    (\mathtt{t}\vlbin{\mathbf{a}}\mathtt{f})}
{\vls(\mathtt{f}\vlbin{\vlor}\mathtt{t})
    \vlbin{\mathbf{a}}
    (\mathtt{t}\vlbin{\vlor}\mathtt{f})}$ represents the
excluded-middle 
$\vlupsmash
{\vlinf{}{}
{\atma\vlor\natma}
{\ttt}}$.
The sub-atoms $ \vls(\mathtt{f}\vlbin{\mathbf{a}}\mathtt{t}) $ and
$ (\mathtt{t}\vlbin{\mathbf{a}}\mathtt{f}) $ stand for the
atoms $ a $ and $ \vlne{a} $, respectively, where $ \mathbf{a} $ is a self-dual 
non commutative relation which obeys the equivalence 
$ \vls(\mathtt{f}\vlbin{\vlor}\mathtt{t})
  \vlbin{\mathbf{a}}
  (\mathtt{t}\vlbin{\vlor}\mathtt{f}) =
  \mathtt{t}\vlbin{\mathbf{a}} \mathtt{t} = \mathtt{t}$.
Instead, the rule
$ \vlinf{}{}
{\vls(\mathtt{f}\vlbin{\vlor}\mathtt{f})
    \vlbin{\mathbf{a}}
    (\mathtt{t}\vlbin{\vlor}\mathtt{t})} 
{\vls(\mathtt{f}\vlbin{\mathbf{a}}\mathtt{t})
    \vlbin{\vlor}
    (\mathtt{f}\vlbin{\mathbf{a}}\mathtt{t})}
$
corresponds to the contraction
$\vlupsmash
{\vlinf{}{}
{\atma}
{\atma\vlor\atma}}$. Under the same representation of $ a $ as before, 
the conclusion represents $ a $ up to the standard equivalences 
$ \mathtt{f}\vlbin{\vlor}\mathtt{f} = \mathtt{f}$
and $ \mathtt{t}\vlbin{\vlor}\mathtt{t} = \mathtt{t} $.


One reason why Subatomic systems are a deep inference formalism is that
they target the representation of a class of logical systems as wide as possible
which may well include self-dual non-commutative logical operators and we know 
that there cannot be analytic and complete Gentzen (linear) proof systems with
self-dual non-commutative connectives in them \cite{DBLP:journals/lmcs/Tiu06}.
Another reason is that, by means of the uniform representation they allow, 
Subatomic systems help to identify sufficient conditions to characterize 
proof systems that enjoy decomposition, i.e. the reorganization of  
contractions inside a proof, and cut-elimination. This is possible because 
Subatomic systems abstract at the right level the proofs of decomposition and 
of cut-elimination that the literature contains in relation to deep inference 
logical systems for classical, modal, linear and sub-structural logics.

Very briefly, deep inference looks at deductive processes as rewriting 
procedures where rules apply to an arbitrary depth in the syntax tree of 
formulas. This is equivalent to saying that deep inference logical systems 
compose derivations and formulas exactly with the same set of logical 
connectives. Subatomic systems witness how effective  the reduction can be of 
syntactic bureaucracy that follows from the deep inference approach to proof 
theory  to get closer to the semantic nature of proof and proof normalization.
An informative survey about deep inference is \cite{Gugl:14:Deep-Inf:fj}.
An up-do-date information about its literature is \cite{Gugl::Deep-Inf:uq}.

This paper introduces Subatomic systems-1.1 
(Section~\ref{section:Generalizing subatomic systems}), a slight generalization 
of the original Subatomic system in 
\cite{Tubella2016,Tubella:2018:SPS:3176362.3173544} that we dub as version 1.0, 
for easiness of reference. Version 1.1 relaxes and widen the properties 
that the sub-atoms of version 1.0 can satisfy while maintaining the use 
of the medial shape as the only inference principle. As effect of the 
generalization, the formulas of 
Subatomic systems 1.1 build also on variables. Hence, we can introduce 
$ \mathsf{P} $ (Section~\ref{section:The system P}.)
We show that $ \mathsf{P} $ is sound and complete with respect to the 
clone at the top of Post's Lattice (Section~\ref{section:Subatomic systems and 
Post's Lattice}.) I.e. $ \mathsf{P} $ proves all and only the tautologies that 
contain conjunctions, disjunctions and the self-dual projections $ \pi_0$ and 
$ \pi_1 $. So, $ \mathsf{P} $ extends Propositional logic without any 
encoding of its atoms as sub-atoms of $ \mathsf{P} $.
We also prove that the cut and other rules are admissible for a specific 
fragment of $ \mathsf{P} $ (Section~\ref{section:Splitting for Pdownarrow}.)
The proof is a corollary of the same property that we prove for version 1.1 and 
which extends the one for version 1.0 (Section~\ref{section:Splitting in 
subatomic systems-1.1}.)

The existence of $ \mathsf{P} $ shows that the logical principles 
underpinning  Subatomic systems also apply outside the sub-atomic level which 
they are conceived to work at. We reinforce this idea by introducing the set 
$\mathsf{R_{23}} $ of medial shapes (Section~\ref{section:conclusion}.)
The formulas that occur in the rules of $ \mathsf{R_{23}} $ belong to the union 
of the two clones $ \mathsf{C_2} $ and $ \mathsf{C_3} $ of Post's Lattice 
\cite{Post1941}. Both $ \mathsf{C_2} $ and $ \mathsf{C_3} $ are two of the five 
maximal clones strictly contained in $ \mathsf{C_1} $. The logical operators 
that build the formulas of $ \mathsf{C_2} $ and $ \mathsf{C_3} $ are strongly 
interrelated but the satisfiability problem for $ \mathsf{C_2} $ is in P-Time 
while the one for $ \mathsf{C_3} $ is NP-Time complete. That $ \mathsf{R_{23}} 
$ can be a Subatomic system-1.1 is still an open question. The conjecture is 
that we need a further extension of Subatomic systems to prove a 
cut-elimination for a system with $ \mathsf{R_{23}} $ as its core. The 
relevance of $ \mathsf{R_{23}}$ is twofold. 
On one side, it can help focusing on proof-theoretical properties that 
highlight how and when the phase transition from the satisfiability in P-Time 
to the satisfiability in the class of NP-Time complete problems occurs.
On the other, the way we obtain $\mathsf{R_{23}} $ strongly suggests that 
Subatomic systems can be viewed as a framework where looking for grammars 
that follow a very regular pattern able to generate possibly interesting 
logical systems, so contributing to the so called systematic proof theory 
\cite{DBLP:conf/csl/CiabattoniST09}. The side effect would be that the larger 
will be the class of interesting logical systems that we can generate by means 
of Subatomic system, the clearer the reason could be why the medial scheme rule 
is so pervasive, something that, so far, has no a priory convincing explanation.

\section{Subatomic systems-1.1}
\label{section:Generalizing subatomic systems}
We generalize Subatomic systems-1.0 
\cite{Tubella2016,Tubella:2018:SPS:3176362.3173544} to Subatomic systems-1.1.

\vllinearfalse
\noindent
\begin{definition}[Subatomic systems-1.1]
\label{definition:A possible generalization of Subatomic systems}
Let $ \mathcal{U} $ be a denumerable set of constants $ t, u, v, \ldots $.
Let $ \mathcal{V} $ be a denumerable set of variables $ x, y, w, \ldots $.
Let $ \mathcal{R} $ be a denumerable set of symbol relations
$ \alpha, \beta, \ldots $ and let $ \prec\ \subseteq \mathcal{R}^2$ be a 
partial order among the symbols in $ \mathcal{R} $.
Let $ \mathcal{F} ::= \mathcal{U} \mid \mathcal{V} \mid \mathcal{F}\, 
\mathcal{R}\, \mathcal{F} $ 
generate formulas $ A, B, C,\ldots $.
Let
$\vlne{(\ \ )}:(\mathcal{U}\rightarrow\mathcal{U})\cup
(\mathcal{V}\rightarrow\mathcal{V})\cup
(\mathcal{R}\rightarrow\mathcal{R})$ be an involutive \emph{negation}: 
\begin{align*}
&&\vlne{A} & = 
\begin{cases}
\ \vlne{u} & \textrm{if } A = u
\textrm{ and } u \in \mathcal{U}
\\
\ \vlne{x} & \textrm{if } A = x
\textrm{ and } x \in \mathcal{V}
\\
\ \vlne{B} \vlbin{\vlne{\alpha}} \vlne{C} 
& \textrm{if } A = B \vlbin{\alpha} C
\textrm{ and } B \vlbin{\alpha} C \in \mathcal{F}
\end{cases}
\enspace. 
\end{align*}
Fixed $ n\in\mathbb{N} $, let $ =\ \subseteq \mathcal{F}^2 $ be the least 
congruence on $ \mathcal{F} $ generated by any subset $ E_1 = F_1, \ldots, E_n 
= F_n$ of axioms taken among following:
\begin{align}
\label{align: associativity ax}
(A\vlbin{\alpha}B)\vlbin{\alpha}C & = A\vlbin{\alpha}(B\vlbin{\alpha}C) 
&& (A, B, C\in\mathcal{F}
   ,\ \alpha\in\mathcal{R})\\
\label{align: commutativity ax}
A\vlbin{\alpha}B  & = B\vlbin{\alpha}A  
&& (A, B \in\mathcal{F}
   ,\ \alpha\in\mathcal{R})\\
\label{align:unitary ax}
A\vlbin{\alpha}u_{\alpha} & = A = u_{\alpha}\vlbin{\alpha}A      
&& (A\in\mathcal{F}
   ,\ u_{\alpha}\in\mathcal{U}
   ,\ \alpha\in\mathcal{R})\\
\label{align:right weakening ax}
A\vlbin{\alpha}B  & = A           
&& (A, B \in\mathcal{F}
   ,\ \alpha\in\mathcal{R})\\
\label{align:left weakening ax}
B\vlbin{\alpha}A  & = A                                 
&& (A, B \in\mathcal{F}
   ,\ \alpha\in\mathcal{R})
\\
\label{align:constant constant assignment  ax}
t\vlbin{\alpha}u  & = v
&& (t, u, v\in\mathcal{U}
   ,\ \alpha\in\mathcal{R})\\
\label{align:variable constant assignment  ax}
x\vlbin{\alpha}y  & = z
&& (x, y, z\in\mathcal{V}
   ,\ \alpha\in\mathcal{R})\\
\label{align:variable variable constant assignment ax}
x\vlbin{\alpha}y  & = u
&& (x, y\in\mathcal{V}, u\in\mathcal{U}
   ,\ \alpha\in\mathcal{R})\\
\label{align:constant equation ax}
t                 & = u
&&  (t, u\in\mathcal{U})
\enspace.
\end{align}
\noindent
A Subatomic system-1.1 $ \mathsf{S} $ on $ \mathcal{F} $, $ \mathcal{R}, \prec 
$ and $ = $ has all and only the instances of the following schemes:
\begin{center}
{\setlength{\extrarowheight}{1.5em}%
    \begin{tabular}{r||c|c}
        & \textbf{Down-rules} & \textbf{Up-rules}\\\hline\hline
        \textbf{Splitting}   &
        $\vlinf{}{(\alpha \prec \beta)}
        {(A \vlbin{\alpha} C) \vlbin{\beta} \left(B \vlbin{\vlne{\alpha}}
        D\right)}
        {(A \vlbin{\beta} B)      \vlbin{\alpha} (C \vlbin{\beta} D)}
        $ &
        $\vlinf{}{(\alpha \succ \beta)}
        {(A \vlbin{\beta} B) \vlbin{\alpha}      (C \vlbin{\beta} D)}
        {(A \vlbin{\alpha} C) \vlbin{\beta} \left(B 
        \vlbin{\vlne{\alpha}} D\right)}
        $
        \\[1.5em]\hline
        \textbf{Contractive} &
        $\vlinf{}{(\alpha \prec \beta)}
        {(A \vlbin{\beta} C) \vlbin{\alpha} (B \vlbin{\beta} D)}
        {(A \vlbin{\alpha} B)      \vlbin{\beta} (C \vlbin{\alpha} D)}
        $ &
        $\vlinf{}{(\alpha \succ \beta)}
        {(A \vlbin{\alpha} B) \vlbin{\beta} (C \vlbin{\alpha} D)}
        {(A \vlbin{\beta} C) \vlbin{\alpha} (B \vlbin{\beta} D)}
        $
        \\[1.5em]\hline
        \textbf{Equational} &
        $\vliqf{}{}
        {F_i}
        {E_i}
        $ &
        $\vliqf{}{}
        {\vlne{F_i}}
        {\vlne{E_i}}
        $
\end{tabular}}
\end{center}
    \qedh
\end{definition}

Like in \cite{Tubella2016,Tubella:2018:SPS:3176362.3173544}, the role of 
Definition~\ref{definition:A possible generalization of Subatomic 
systems} is to delineate the formal framework we are going to work in. The 
constraints on the framework are very lax. It should not surprise how simple is 
to think of semantically meaningless instances of Subatomic systems where, for 
example, the two propositional constants $ \bTT $ (true) and $ \bFF $ (false) 
exist and are equated by an instance of~\eqref{align:constant equation ax}.

The framework we delineate is slightly more general than the one in 
\cite{Tubella2016,Tubella:2018:SPS:3176362.3173544}. The language $ \mathcal{F} 
$ contains variables of $ \mathcal{V} $ and the set of axioms is extended in 
two directions. Axioms~\eqref{align:right weakening ax}, \eqref{align:left 
weakening  ax} admit the existence of relations that erase structure. 
Axioms~\eqref{align:variable constant assignment  ax} and~\eqref{align:variable 
variable constant assignment ax} allow the existence of relations among 
constants and variables. This extends the proof theoretical 
technology of Subatomic systems-1.0 outside its intrinsic sub-atomic nature.

\begin{notation*}
Let $ \mathsf{S} $  be a Subatomic system-1.1 with formulas $ \mathcal{F} $ 
built on the symbols in $ \mathcal{R} $. Let $ \prec $ be the order relation on 
$ \mathcal{R}^2 $. 
A context $ S\!\Set{\, } $ is a formula $ A\in\mathcal{F} $ with any of its 
sub-formulas, possibly $ A $ itself, erased. In the last case $ S\!\Set{\, } $ 
is $\Set{\, } $.
A relation $ \alpha $ of $ \mathsf{S} $ is \emph{unitary} if it enjoys 
axiom~\eqref{align:unitary ax}.
A relation $ \alpha $ is a \emph{right weakening} if it 
enjoys~\eqref{align:right weakening ax} and is a \emph{left weakening} 
if~\eqref{align:left weakening ax} holds for it.
A relation $ \alpha\in\mathcal{R} $ is \emph{strong} if no $\beta\in\mathcal{R} 
$ exists such that $ \beta \prec \alpha $.
A relation $ \alpha\in\mathcal{R} $ is \emph{weak} if no $ \beta\in\mathcal{R} 
$ exists such that $ \alpha \prec \beta $.
The map $\vlne{(\ \ )}$ is $ \prec$-\emph{consistent} if a strong 
$\alpha\in\mathcal{R} $ implies that $ \vlne{\alpha} $ is weak, and vice versa.
A \emph{derivation} $ {\vlder{\mathcal{D}}{\mathsf{S}}{B}{A}} $ of $\mathsf{S}$ 
from $ A $ to $ B $ is any obvious concatenation of rules instances of $ 
\mathsf{S} $.
\qedh
\end{notation*}

\begin{remark}
    Strong relations are defined as minimal elements of the partial order
    $ \prec\ \subseteq \mathcal{R}^2 $. Dually, weak relations
    are maximal elements. We share this terminological choice with 
    \cite{Tubella:2018:SPS:3176362.3173544}. The justification is 
    semantical. A relation is strong if its truth implies the 
    truth of a weaker one. For example, the classical conjunction is strong 
    and the classical disjunction weak.
    \qedh
\end{remark}

\begin{proposition}[Excluded middle]
    \label{proposition:subatominc systems 1.1 excluded middle}
    Let $ \mathsf{S} $ be a Subatomic system-1.1 with $ = $ as its 
    equational theory.
    Let $ \alpha\in\mathcal{R} $ be \emph{strong}.
    Let the following instances 
    of~\eqref{align:constant constant assignment  ax} 
    and~\eqref{align:variable variable constant assignment ax} 
    hold in $ \mathsf{S} $:
    \begin{align}
    \label{align:generic excluded middle 01}
    && v \vlbin{\vlne{\alpha}} \vlne{v} & = u_{\alpha} 
      && (\forall v\in\mathcal{U}) \\
    \label{align:generic excluded middle 02}
    && u_\alpha \vlbin{\gamma} u_\alpha & = u_{\alpha} 
      && (\forall \gamma \in\mathcal{R}. \gamma \prec \vlne{\alpha})\\
    \label{align:generic excluded middle 00}
    && x \vlbin{\vlne{\alpha}} \vlne{x} & = u_{\alpha} 
      && (\forall x\in\mathcal{V})
    \end{align}
    where $ u_{\alpha} $ is a single and distinguished element of 
    $ \mathcal{U} $. The rule 
    $ 
    \upsmash{\downsmash{\vlinf{}{}{A\vlbin{\vlne{\alpha}}\vlne{A}}{u_{\alpha}}}}
     $\ is derivable, for every $ A\in\mathcal{F} $.
\end{proposition}
\begin{proof}
    The proof is by induction on the structure of $ A $. 
    The two base cases with  $ A = x $ or $ A = v $ hold
    because~\eqref{align:generic excluded middle 00} and~\eqref{align:generic 
    excluded middle 01} hold in the given $ \mathsf{S} $.
    Let $ A $ be $ A_0 \vlbin{\gamma} A_1 $ where, we underline, $ \gamma  $ 
    can also be $ \alpha $ itself. Moreover, $ \alpha $ strong implies that
    $ \vlne{\alpha} $ weak. Then\
    $
    {\vlderivation{
            \vliq{}{}
            {(A_0\vlbin{\gamma}A_1)\vlbin{\vlne{\alpha}}
             (\vlne{A_0\vlbin{\gamma}A_1})}
            {\vlin{}{(\gamma\prec\vlne{\alpha})}
            {(A_0\vlbin{\gamma}A_1)\vlbin{\vlne{\alpha}}
             (\vlne{A_0}\vlbin{\vlne{\gamma}}\vlne{A_1})}
            {\vlin{}{\textrm{\scriptsize inductive hypothesys}}
            {(A_0\vlbin{\vlne{\alpha}}\vlne{A_0})\vlbin{\gamma}
             (A_1\vlbin{\vlne{\alpha}}\vlne{A_1})}
            {\vliq{\eqref{align:generic excluded middle 02}}{}
            {u_{\alpha}\vlbin{\gamma}u_{\alpha}}
            {\vlhy{u_{\alpha}}}}}}
    }} $.
\end{proof}
\noindent
Proposition~\ref{proposition:subatominc systems 1.1 excluded middle} 
justifies the following:
\begin{definition}[Unit]
\label{definition:Unity]}
The constant $ u_\alpha 
\in\mathcal{U} $ is a \emph{unit} if it enjoys 
axioms~\eqref{align:generic excluded middle 01}, 
\eqref{align:generic excluded middle 02} 
and~\eqref{align:generic excluded middle 00}.
\qedh
\end{definition}

\begin{proposition}[Contraction]
    \label{proposition:subatominc systems 1.1 idempotence}
    Let $ \mathsf{S} $ be a Subatomic system-1.1 with $ = $ as its equational 
    theory.  Let $ \beta\in\mathcal{R} $ be \emph{weak}. Let the 
    following instances
    of~\eqref{align: associativity ax},
    \eqref{align: commutativity ax},
    \eqref{align:constant constant assignment  ax} 
    and~\eqref{align:variable constant assignment  ax} 
    hold in $ \mathsf{S} $:
    \begin{align}
    \label{align:generic idempotence}
    && (A \vlbin{\beta} B)\vlbin{\beta} C  & =  
       A \vlbin{\beta} (B\vlbin{\beta} C)  && (\forall A, B, C\in\mathcal{F}) 
       \\
    \label{align:generic idempotence 0}
    && A         \vlbin{\beta}  B  & =  
    B         \vlbin{\beta}  A  && (\forall A, B\in\mathcal{F}) \\
    \label{align:generic idempotence 00}
    && v         \vlbin{\beta}  v  & = v && (\forall v\in\mathcal{U}) \\
    \label{align:generic idempotence 01}
    && x         \vlbin{\beta}  x  & = x && (\forall x\in\mathcal{V}) 
    \enspace .
    \end{align}
    The rule $\upsmash{\downsmash{\vlinf{}{}{A}{A\vlbin{\beta} A}}} $ is 
    derivable, for 
    every $ 
    A\in\mathcal{F} $.
\end{proposition}
\begin{proof}
    The proof is by induction on the structure of $ A $. The base cases $ A = x 
    $ and $ A = v $ holds because~\eqref{align:generic idempotence 00} 
    and ~\eqref{align:generic idempotence 01} hold in the given $ \mathsf{S} $.
    Let $ A $ be $ A_0\vlbin{\gamma} A_1 $, for any $ \gamma\prec\beta $. Then 
    the following derivation\
    $ \vlderivation{
      \vlin{}{\textrm{\scriptsize inductive hypothesys}}
      {A_0\vlbin{\gamma}A_1}
      {\vlin{}{(\gamma\prec\beta)}
      {(A_0\vlbin{\beta}A_0)\vlbin{\gamma}(A_1\vlbin{\beta}A_1)}
      {\vlhy{(A_0\vlbin{\gamma}A_1)\vlbin{\beta}(A_0\vlbin{\gamma}A_1)}}}
    } $ exists.
    Finally, let $ A $ be $ A_0\vlbin{\beta} A_1 $. Then\
    $ 
      {\vlderivation{
      \vlin{}{\textrm{\scriptsize inductive hypothesys}}
      {A_0\vlbin{\beta}A_1}
      {\vliq{\eqref{align:generic idempotence}, \eqref{align:generic 
      idempotence 0}}{}
      {(A_0\vlbin{\beta}A_0)\vlbin{\beta}(A_1\vlbin{\beta}A_1)}
      {\vlhy{(A_0\vlbin{\beta}A_1)\vlbin{\beta}(A_0\vlbin{\beta}A_1)}}}}
    } $ exists.
\end{proof}
Propositions~\ref{proposition:subatominc systems 1.1 excluded middle} 
and~\ref{proposition:subatominc systems 1.1 idempotence} say 
that the medial shape is an invariant of two inference mechanisms.
One is ``\emph{Splitting}'' or, dually, ``\emph{annihilation}''.
It distributes negation. So, the proofs of a Subatomic system-1.1 can 
start from units which split into a pair of structures that annihilate 
each other. The other is 
``\emph{Contraction}'' or, dually, ``\emph{sharing}''. It distributes 
sub-formulas with the goal of identifying two occurrences of the same 
formula into a single one. This is a consequence of a step-wise  
deductive process that reduces the global identification to the 
identification on constants or variables only.

\begin{fact}[Equation derivations]
    \label{fact:Equation derivations}
    Let $ \mathcal{D} $ be a derivation that only contains equation rules
    of a given Subatomic system-1.1 $ \mathsf{S} $. We can obtain derivations 
    of $ \mathsf{S} $
    from $ \mathcal{D} $ in two steps: (i) negating every of formula of $ 
    \mathcal{D} $, 
    (ii) flipping $ \mathcal{D} $ up-side down.
    \qedh
\end{fact}
\vllinearfalse
\section{The Subatomic system-1.1 $ \mathsf{P}$}
\label{section:The system P}
We introduce the instance $ \mathsf{P} $ of Subatomic systems-1.1 which 
we could not see how to obtain as an instance of Subatomic systems-1.0 
\cite{Tubella2016,Tubella:2018:SPS:3176362.3173544}.

\begin{definition}[Formulas of $ \mathsf{P} $]
\label{definition: formulas of P}
Let $ \mathcal{F}_{\mathsf{P}} $ be the language of formulas generated by:
\begin{align*}
&& 
A, B & ::= \bTT \mid \bFF \mid \mathcal{V}_{\mathsf{P}} \mid 
\vlne{\mathcal{V}}_{\mathsf{P}} \mid 
A \vlan B \mid A \vlbin{\pi_0} B \mid A \vlbin{\pi_1} B \mid A \vlor B
\enspace .
\end{align*}
The set $ \mathcal{V}_{\mathsf{P}} $ contains the variables $ x, y, z, \ldots $
and $ \vlne{\mathcal{V}}_{\mathsf{P}} $ their negations.
Both $ \pi_0 $ and $ \pi_1 $ stand for the \emph{self-dual} projections on 
first or second argument, respectively.
\qedh
\end{definition}

\begin{definition}[Order relation among the relations of $ \mathsf{P} $]
\label{definition: Order relation among the relations of P}
The operator $ \vlan $ is \emph{strong}, $ \vlor $ is \emph{weak} and every $ 
\pi_i $ is in between. i.e. 
$ A \vlan B \prec_{\mathsf{P}} A \vlbin{\pi_0} B, A \vlbin{\pi_1} B 
\prec_{\mathsf{P}} A \vlor B$.
\qedh
\end{definition}
\noindent
The order relation of Definition~\ref{definition: Order relation among the 
relations of P} originates from the following lattice which
pointwise sorts the boolean functions it contains under the assumption that 
$ \bFF $ is smaller than $ \bTT $:
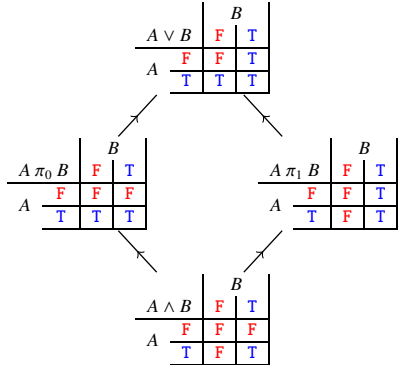
\begin{center}
\begin{tikzpicture}[scale=.6]
\begin{pgfonlayer}{nodelayer}

\node [style=none,scale=0.7] (pi1)    at (-2.75, 0.0) {
    \vlstore{A \vlbin{\pi_0} B}
    \begin{tabular}{cc|c|c|c|l}
    & & \multicolumn{2}{ c| }{$ B $}        \\ 
    \multicolumn{2}{c|}{$ \vlread $} & $\bFF$ & $\bTT$  \\ 
    \hline
    \multicolumn{1}{ c  }{\multirow{2}{*}{$ A $} } 
    & \multicolumn{1}{ c| }{$\bFF$} & $\bFF$ & $\bFF$  \\ \cline{2-4}
    \multicolumn{1}{ c  }{} 
    & \multicolumn{1}{ c| }{$\bTT$} & $\bTT$ & $\bTT$   \\ \cline{2-4}
    \end{tabular}
};

\node [style=none,scale=0.7] (or)     at ( 0.0 ,-3.0) {
    \vlstore{A \vlbin{\vlan} B}
    \begin{tabular}{cc|c|c|c|l}
    & & \multicolumn{2}{ c| }{$ B $}        \\ 
    \multicolumn{2}{c|}{$ \vlread $} & $\bFF$ & $\bTT$  \\ 
    \hline
    \multicolumn{1}{ c  }{\multirow{2}{*}{$ A $} } 
    & \multicolumn{1}{ c| }{$\bFF$} & $\bFF$ & $\bFF$  \\ \cline{2-4}
    \multicolumn{1}{ c  }{} 
    & \multicolumn{1}{ c| }{$\bTT$} & $\bFF$ & $\bTT$   \\ \cline{2-4}
    \end{tabular}
};

\node [style=none,scale=0.7] (and)    at ( 0.0 , 3.0) {
    \vlstore{A \vlbin{\vlor} B}
    \begin{tabular}{cc|c|c|c|l}
    & & \multicolumn{2}{ c| }{$ B $}        \\ 
    \multicolumn{2}{c|}{$ \vlread $} & $\bFF$ & $\bTT$  \\ 
    \hline
    \multicolumn{1}{ c  }{\multirow{2}{*}{$ A $} } 
    & \multicolumn{1}{ c| }{$\bFF$} & $\bFF$ & $\bTT$  \\ \cline{2-4}
    \multicolumn{1}{ c  }{} 
    & \multicolumn{1}{ c| }{$\bTT$} & $\bTT$ & $\bTT 
    $   \\ \cline{2-4}
    \end{tabular}
};

\node [style=none,scale=0.7] (pi0)    at ( 2.75, 0.0) {
    \vlstore{A \vlbin{\pi_1} B}
    \begin{tabular}{cc|c|c|c|l}
    & & \multicolumn{2}{ c| }{$  B $}        \\ 
    \multicolumn{2}{c|}{$ \vlread $} & $\bFF$ & $\bTT$  \\ 
    \hline
    \multicolumn{1}{ c  }{\multirow{2}{*}{$ A $} } 
    & \multicolumn{1}{ c| }{$\bFF$} & $\bFF$ & $\bTT$  \\ \cline{2-4}
    \multicolumn{1}{ c  }{} 
    & \multicolumn{1}{ c| }{$\bTT$} & $\bFF$ & $\bTT$   \\ \cline{2-4}
    \end{tabular}
};
\end{pgfonlayer}
\begin{pgfonlayer}{edgelayer}
\draw [style=arrow, color=black]  (or)  to (pi1) ;
\draw [style=arrow, color=black]  (or)  to (pi0) ;
\draw [style=arrow, color=black]  (pi1) to (and) ;
\draw [style=arrow, color=black]  (pi0) to (and);
\end{pgfonlayer}
\end{tikzpicture}
\end{center}

\begin{definition}[Negation among formulas of $ \mathsf{P} $]
\label{definition: negation among formulas of P}
For every $ x, A, B\in\mathcal{F}_{\mathsf{P}}$,
let $ \vlne{(\ )} $ be the following involutive and $ \prec_{\mathsf{P}} $-consistent negation:
\begin{align} 
\label{align:axiom negation of constants}
&&\vlne{\bTT}  &= \bFF \\
&&\nonumber
\vlne{\bFF}  &= \bTT \\
&&\nonumber
\vlne{\vlne{x}}  &= x && (\forall x\in\mathcal{V})\\
&&\label{align:axiom de morgan or and}
\vlne{A \vlor B} &= \vlne{A} \vlan \vlne{B} && (\forall A, B\in\mathcal{F}) \\
&&\nonumber
\vlne{A \vlan B} &= \vlne{A} \vlor \vlne{B} && (\forall A, B\in\mathcal{F})\\
&&\label{align:axiom de morgan pi}
\vlne{A \vlbin{\pi_i} B} &= \vlne{A} \vlbin{\pi_i} \vlne{B} 
&& (\forall A, B\in\mathcal{F} \textrm{ and } i\in\Set{0, 1}) 
\enspace .
\end{align} 
\qedh
\end{definition}
\noindent
Axiom~\eqref{align:axiom de morgan pi} sets $ \pi_0 $ and $ \pi_1 $ to be
self-dual operators like the boolean functions they represent.

\begin{definition}[Congruence on formulas of $ \mathsf{P} $]
\label{definition: Congruence among formulas of P}
    Let $\bFF$ be the unit $ u_{\vlor} $ of $ \vlor $
    and $\bTT$ the unit $ u_{\vlan} $ of $ \vlan $.
    For every $ A, B $ and $ C $ in $ \mathcal{F}_{\mathsf{P}} $,
    let $ =_{\mathsf{P}} $ be the congruence that the following axioms induce:
    \begin{align} 
    \label{align:axiom Associativity}
    (A \vlbin{\alpha} B) \vlbin{\alpha} C  &= A \vlbin{\alpha} (B \vlbin{\alpha} C)
    && (\forall A, B, C \in\mathcal{F} 
        \textrm{ and }\alpha\in\Set{\pi_0, \pi_1, \vlor, \vlan}) \\
    \label{align:axiom Commutativity}
    A \vlbin{\alpha} B &= B \vlbin{\alpha} A
    && (\forall A, B \in\mathcal{F} 
       \textrm{ and }\alpha\in\Set{\vlor, \vlan}) \\
    \label{align:axiom unit of or}
    A    \vlor \bFF &= A && (\forall A \in\mathcal{F})\\
    \label{align:axiom right weakening pi0}
    A \vlbin{\pi_0} B &= A && (\forall A, B \in\mathcal{F})\\
    \label{align:axiom right weakening pi1}
    A \vlbin{\pi_1} B &= B && (\forall A, B \in\mathcal{F}) \\
    u    \vlan \bFF &= \bFF && (u\in\{\bFF, \bTT \}) \\
    \label{align:axiom constant excluded middle p}
    u \vlor \vlne{u} &= \bTT && (u\in\{\bFF, \bTT \}) \\
    \label{align:axiom excluded middle p}
    x \vlor \vlne{x} &= \bTT && (\forall x\in\mathcal{V}_{\mathsf{P}}) \\
    \label{align:axiom  atomic constant idempotence P}
    u \vlor u &= u && (u\in\{\bFF, \bTT \}) \\
    \label{align:axiom  atomic idempotence P}
    x \vlor x &= x && (\forall x\in\mathcal{V}_{\mathsf{P}})
    \enspace .
    \end{align}
\qedh
\end{definition}
\noindent
Definition~\ref{definition: Congruence among formulas of P} gives the least set 
of axioms. The missing ones can be derived by negation.

\begin{definition}[System $ \mathsf{P} $]
\label{definition: Subatomic system 1.1 P}
$ \mathsf{P}$ contains the rules:
\begin{center}
  \newcommand{\localF}     {\pi_j}
  \newcommand{\localGlb}   {\vlor}
  \newcommand{\localLub}   {\vlan}
  \newcommand{\localnotF}  {\pi_j}
  \newcommand{\localnotGlb}{\vlan}
  \newcommand{\localnotLub}{\vlor}
  \resizebox{.85\textwidth}{!}{
	\begin{tabular}{cccc}
        $\vlinf{ai_j\downarrow}{j\in\Set{0,1}} 
		{(A \vlbin{\localF} C)\vlbin{\localGlb}(B \vlbin{\localnotF} D)}
		{(A \vlbin{\localGlb} B)\vlbin{\localF}(C \vlbin{\localGlb} D)}$
		&\qquad\qquad\qquad&
		$\vlinf{ai_j\uparrow}{j\in\Set{0,1}} 
		{(A \vlbin{\localnotGlb} B)\vlbin{\localnotF}(C \vlbin{\localnotGlb} D)}
		{(A \vlbin{\localnotF} C)\vlbin{\localnotGlb}(B \vlbin{\localF} D)}$
		\end{tabular}
	} 
\end{center}
\begin{center}
    \newcommand{\localF}     {\pi_j}
    \newcommand{\localGlb}   {\vlor}
    \newcommand{\localLub}   {\vlan}
    \newcommand{\localnotF}  {\pi_j}
    \newcommand{\localnotGlb}{\vlan}
    \newcommand{\localnotLub}{\vlor}
    \resizebox{.85\textwidth}{!}{
        \begin{tabular}{cccc}
            $\vlinf{s\downarrow}{\phantom{i\in\Set{0,1}}}
            {(A \vlbin{\localLub} C)\vlbin{\localGlb}(B \vlbin{\localnotLub} D)}
            {(A \vlbin{\localGlb} B)\vlbin{\localLub}(C \vlbin{\localGlb}    D)}$
            &\qquad\qquad\qquad&
            $\vlinf{s\uparrow}{\phantom{i\in\Set{0,1}}}
            {(A \vlbin{\localnotGlb} B)\vlbin{\localnotLub}(C \vlbin{\localnotGlb} D)}
            {(A \vlbin{\localnotLub} C)\vlbin{\localnotGlb}(B \vlbin{\localLub}    D)}$
            \\\\
            $\vlinf{m\downarrow}{\phantom{i\in\Set{0,1}}}
            {(A \vlbin{\localGlb} C)\vlbin{\localLub}(B \vlbin{\localGlb} D)} 
            {(A \vlbin{\localLub} B)\vlbin{\localGlb}(C \vlbin{\localLub} D)}$
            &\qquad\qquad\qquad&
            $\vlinf{m\uparrow}{\phantom{i\in\Set{0,1}}}
            {(A \vlbin{\localnotLub} B)\vlbin{\localnotGlb}(C \vlbin{\localnotLub} D)} 
            {(A \vlbin{\localnotGlb} C)\vlbin{\localnotLub}(B \vlbin{\localnotGlb} D)}$
            \\\\
            $\vlinf{c_j\downarrow}{j\in\Set{0,1}}
            {(A \vlbin{\localGlb} C)\vlbin{\localF  }(B \vlbin{\localGlb} D)} 
            {(A \vlbin{\localF  } B)\vlbin{\localGlb}(C \vlbin{\localF  } D)}$
            &\qquad\qquad\qquad&
            $\vlinf{c_j\uparrow}{j\in\Set{0,1}}
            {(A \vlbin{\localnotF  } B)\vlbin{\localnotGlb}(C \vlbin{\localnotF  } D)} 
            {(A \vlbin{\localnotGlb} C)\vlbin{\localnotF  }(B \vlbin{\localnotGlb} D)}$
            &
        \end{tabular}
    } 
\end{center}
with formulas of $ \mathcal{F}_{\mathsf{P}} $ 
(Definition~\ref{definition: formulas of P})
taken up to both $ =_{\mathsf{P}} $ (Definition~\ref{definition: Congruence 
among formulas of P}) and the negation in Definition~\ref{definition: negation 
among formulas of P}, with $ \vlan , \pi_0, \pi_1 $ and $ \vlor $ 
ordered under $ \prec_{\mathsf{P}} $ (Definition~\ref{definition: Order 
relation among the relations of P}.)
\qedh
\end{definition}
\noindent
So, $ \mathsf{P} $ is a Subatomic system-1.1 because its formalization fits in 
the framework of Definition~\ref{definition:A possible generalization of 
Subatomic systems}. Hence, Proposition~\ref{proposition:subatominc systems 1.1 
excluded middle}, 
axioms~\eqref{align:axiom Associativity},
\eqref{align:axiom Commutativity}, 
\eqref{align:axiom  atomic constant idempotence P} 
and~\eqref{align:axiom  atomic idempotence P}, 
and the rules $ ai_0\!\!\downarrow, ai_1\!\!\downarrow, s\!\!\downarrow$ imply:
\begin{corollary}[Excluded middle in $ \mathsf{P} $]
    \label{lemma:Generic excluded middle in P}
    For every $ A\in\mathcal{F}_{\mathsf{P}} $, the rule 
    $ \vlupsmash{\vldownsmash{\vldownsmash{\vlinf{}{}{A\vlor \vlne{A}}{\bTT}}}} 
    $\ \
    is derivable.
\end{corollary}
\noindent
Moreover, Proposition~\ref{proposition:subatominc systems 1.1 idempotence}, 
axiom~\eqref{align:axiom unit of or},
\eqref{align:axiom right weakening pi0},
\eqref{align:axiom right weakening pi1},
\eqref{align:axiom constant excluded middle p}
and~\ref{align:axiom excluded middle p},
and the rules 
$m\!\!\downarrow, c_0\!\!\downarrow, c_1\!\!\downarrow$ imply:
\begin{corollary}[Idempotence in $ \mathsf{P} $]
    \label{lemma:Generic idempotence in P}
    For every $ A\in\mathcal{F}_{\mathsf{P}} $, the rule 
    $ \vlupsmash{\vlinf{}{}{A}{A\vlor A}} $\ is derivable.
\end{corollary}

\begin{remark}
As far as we can see, $ \mathsf{P} $ cannot be a Subatomic 
system-1.0, in accordance with Definition~2.5 in \cite[page 10]{Tubella2016} 
and \cite[page 6]{Tubella:2018:SPS:3176362.3173544}.
The axiom scheme (3) of those two works classifies every \emph{unitary} 
relation $ \alpha $ as one for which we have:
\begin{align}
\label{align:unit of subatomic systems}
&& A \vlbin{\alpha} u & = A = u \vlbin{\alpha} A 
\enspace .
\end{align}
However, the natural behavior of the relations $ \pi_0 $ and $ \pi_1 $ of $ 
\mathsf{P} $ 
is given by~\eqref{align:axiom right weakening pi0} 
and~\eqref{align:axiom right weakening pi1}, instances 
of~\eqref{align:right weakening ax} and~\eqref{align:left weakening ax}.
So, they cannot satisfy~\eqref{align:unit of subatomic systems}.
We will see that the weaker behavior of $ \pi_0 $ and $ \pi_1 $ as compared to 
the one of a unitary relation requires to generalize the Splitting theorem 
(Section~\ref{section:Splitting in subatomic systems-1.1}.)
\qedh
\end{remark}

\begin{remark}
There is an aspect of $ \mathsf{P} $ for which we have no convincing a priori 
justification. For every $ j $, the rule $ c_j\!\!\downarrow $ is $ 
ai_j\!\!\downarrow $  flipped up side down.  Currently, we limit to observe 
that this is harmless. Both rules are semantically sound, i.e. the truth of the 
premise implies the one of the conclusion.
\qedh
\end{remark}

\section{Cut-elimination in Subatomic systems-1.1}
\label{section:Splitting in subatomic systems-1.1}
We here adapt the proof of the cut-elimination for Subatomic systems-1.0 
\cite{Tubella2016,Tubella:2018:SPS:3176362.3173544} to version 1.1.

\begin{definition}[Splittable down-fragment]
    \label{definition:splittable subatomic systems-1.1}
    Let $ \mathsf{S}$ be a Subatomic system-1.1.
    Then, $ \mathsf{S}\!\!\downarrow $ is the \emph{Splittable down-fragment}
    of $ \mathsf{S} $ if:
    \begin{enumerate}
        \item $ \mathsf{S}\!\!\downarrow $ contains at least one weak relation;
        \item For every weak relation $ \beta $ in $ \mathsf{S}\!\!\downarrow $ 
        with unit $ u_{\beta}\in\mathcal{U} $
        the following axioms hold:
        \begin{align}
        \label{align:ax6 gen SA}
        \vlne{u_\beta}\vlbin{\alpha}\vlne{u_\beta}& = 
        \vlne{u_{\beta}}                    &&\forall \alpha\in\mathcal{R}. 
        \alpha\prec\beta\\
        \label{align:ax1 gen SA}
        (A \vlbin{\beta} B) \vlbin{\beta} C & = A \vlbin{\beta} (B \vlbin{\beta} C) &&\forall A, B, C\in\mathcal{F}\\
        \label{align:ax2 gen SA}
        A \vlbin{\beta} B                  & = B \vlbin{\beta} 
        A                   &&\forall A, B   \in\mathcal{F}\\
        \label{align:ax4 gen SA}
        u\vlbin{\beta}\vlne{u}              & = 
        \vlne{u_{\beta}}                    &&\forall u\in\mathcal{U}\\
        \label{align:ax3 gen SA}
        A          \vlbin{\beta} u_{\beta}  & = A                                   &&\forall A\in\mathcal{F}\\
        \label{align:ax5 gen SA}
        x\vlbin{\beta}\vlne{x}              & = 
        \vlne{u_{\beta}}                    &&\forall x\in\mathcal{V}
        \enspace ;
        \end{align}
        
        \item 
        $ \mathsf{S}\!\!\downarrow $ contains all and only the splitting and 
        equational down-rules, as in
        Definition~\ref{definition:A possible generalization of Subatomic 
        systems}. So, it does not contain any contractive down-rule.
    \qedh
    \end{enumerate}
\end{definition}

\noindent
Like in \cite{Tubella2016,Tubella:2018:SPS:3176362.3173544}, 
axioms~\eqref{align:ax6 gen SA}, 
\eqref{align:ax1 gen SA} 
and~\eqref{align:ax2 gen SA} are strongly linked to the way that splitting
works. Once decomposed a proof into independent subproofs, they can be composed 
back into a new proof exactly because the here above axioms hold.
Also~\eqref{align:ax4 gen SA} is in 
\cite{Tubella2016,Tubella:2018:SPS:3176362.3173544}.
Instead, both~\eqref{align:ax3 gen SA} and~\eqref{align:ax5 gen SA} are new.

\noindent
Symmetrically to Definition~\ref{definition:splittable subatomic systems-1.1}, 
we can identify the splittable down-fragment.
\begin{definition}[Splittable up-fragment]
    \label{definition:Up-fragment of S}
    Let $ \mathsf{S} $ be a Subatomic system-1.1 with a Splittable down-fragment
    $ \mathsf{S}\!\!\downarrow $ as in Definition~\ref{definition:splittable 
        subatomic systems-1.1}.
    The Splittable up-fragment $ \mathsf{S}\!\!\uparrow $ 
    contains all and only the splitting and equational \emph{up}-rules of 
    $ \mathsf{S} $ that correspond to the rules in $ \mathsf{S}\!\!\downarrow $.
    \qedh
\end{definition}

\begin{definition}[Length of a derivation]
    \label{definition:Length of a derivation}
    Let $ \mathsf{S}\!\!\downarrow $ be a \emph{Splittable} Subatomic system-1.1.
    The \emph{length} $ |\mathcal{D}| $ of a derivation $ \mathcal{D} $ in $ 
    \mathsf{S}\!\!\downarrow $
    counts the number of rules of $ \mathcal{D} $ which do not correspond to 
    the application of any axiom among~\eqref{align:ax1 gen SA},
    \eqref{align:ax2 gen SA}, \eqref{align:ax3 gen SA}, \eqref{align:ax4 gen 
    SA} and~\eqref{align:ax5 gen SA}
    at point 2 of Definition~\ref{definition:splittable subatomic systems-1.1}.
    I.e. axioms that involve \emph{weak relations} do not contribute to the 
    growth of the dimension of a derivation, \emph{while~\eqref{align:ax6 gen 
    SA} does.}
    \qedh
\end{definition}

\begin{lemma}[Atomic deduction for Splittable Subatomic systems-1.1]
    \label{lemma:Atomic Deduction for generalized subatomic systems}
    Let $ \beta $ be a weak relation with unit $ u_{\beta} $ of a Splittable 
    Subatomic system-1.1
    $ \mathsf{S}\!\!\downarrow$. For every $ u\in\mathcal{U}$ and 
    $ B\in\mathcal{F}$,
    if 
    $ \vldownsmash
      {\vlder{\mathcal{P}}{ \mathsf{S}\downarrow}{u \vlbin{\beta} B}
      {\vlne{u_{\beta}}}} $, 
    then 
    $\vldownsmash
      { \vlder{\mathcal{D}}
      { \mathsf{S}\downarrow}{B}{\vlne{u}}} $ exists.
    Analogously, for every $ x\in\mathcal{V}$,
    if 
    $\vldownsmash
      { \vlder{\mathcal{P}'}
      { \mathsf{S}\downarrow}{x \vlbin{\beta} B}{\vlne{u_{\beta}}}} $, 
    then 
    $\vldownsmash
     { \vlder{\mathcal{D}'}
     { \mathsf{S}\downarrow}{B}{\vlne{x}} }$ exists.
\end{lemma}
\begin{proof}
    The derivation $ \mathcal{D} $ is\
    $ 
      {\vlderivation{
       \vliq{\eqref{align:ax2 gen SA},\eqref{align:ax3 gen SA}}{}
       {B}
       {\vliq{\eqref{align:ax4 gen SA}}{}
       {u_{\beta}\vlbin{\beta}(u_{\beta}\vlbin{\beta} B)}
       {\vlin{s\downarrow}{}
       {(\vlne{u}\vlbin{\vlne{\beta}} u)\vlbin{\beta}(u_{\beta}\vlbin{\beta} B)}
       {\vlde{\mathcal{P}}{ \mathsf{S}\downarrow}
       {(\vlne{u}\vlbin{\beta}u_{\beta})\vlbin{\vlne{\beta}}(u\vlbin{\beta} B)}
       {\vliq{\eqref{align:ax3 gen SA}}{}
       {(\vlne{u}\vlbin{\beta}u_{\beta})\vlbin{\vlne{\beta}}\vlne{u_{\beta}}}
       {\vliq{\eqref{align:ax3 gen SA}}{}
       {\vlne{u}\vlbin{\beta}u_{\beta}}
       {\vlhy{\vlne{u}}}}}}}}
    }} $\ , 
    while $ \mathcal{D}' $ is\
    $ 
      {\vlderivation{
       \vliq{\eqref{align:ax2 gen SA}, \eqref{align:ax3 gen SA}}{}
       {B}
       {\vliq{\eqref{align:ax5 gen SA}}{}
       {u_{\beta}\vlbin{\beta}(u_{\beta}\vlbin{\beta} B)}
       {\vlin{s\downarrow}{}
       {(\vlne{x}\vlbin{\vlne{\beta}} x)\vlbin{\beta}(u_{\beta}\vlbin{\beta} B)}
       {\vlde{\mathcal{P}'}{ \mathsf{S}\downarrow}
       {(\vlne{x}\vlbin{\beta}u_{\beta})\vlbin{\vlne{\beta}}(x\vlbin{\beta} B)}
       {\vliq{\eqref{align:ax3 gen SA}}{}
       {(\vlne{x}\vlbin{\beta}u_{\beta})\vlbin{\vlne{\beta}}\vlne{u_{\beta}}}
       {\vliq{\eqref{align:ax3 gen SA}}{}
       {\vlne{x}\vlbin{\beta}u_{\beta}}
       {\vlhy{\vlne{x}}}}}}}}
    }} $\ . 
\end{proof}

The following theorem strictly generalizes the namesake one in 
\cite{Tubella:2018:SPS:3176362.3173544}.

\begin{theorem}[Shallow splitting]
    \label{theorem:Shallow Splitting for generalized subatomic systems}
    Let $ \beta $ be a weak relation with unit $ u_{\beta} $ in a Splittable 
    Subatomic system-1.1 
    $ \mathsf{S}\!\!\downarrow $. 
    For every $ \alpha\prec \beta $,
    let $ \vlder{\mathcal{P}}{ \mathsf{S}\downarrow}{(A_0 \vlbin{\alpha} A_1) \vlbin{\beta} B}{\vlne{u_{\beta}}} \,$
    be given.
    \begin{enumerate}
        \item 
        If $ \alpha $ is a right weakening,
        then
        $ 
          {\vlder{\mathcal{D}}{ \mathsf{S}\downarrow}{B}{K_0 
          \vlbin{\vlne\alpha} K_1}} $
        exists such that 
        $ 
           {\vlder{\mathcal{Q}_0}{ \mathsf{S}\downarrow}{A_0 \vlbin{\beta} 
           K_0}{\vlne{u_{\beta}}}}\enspace $ exists as well
        and $ |\mathcal{Q}_0|\leq |\mathcal{P}| $.
        If $ \alpha $ is a left weakening, replace $ 1 $ for $ 0 $.
        \item 
        If $ \alpha $ is unitary, then
        $ \vldownsmash
          {\vlder{\mathcal{D}}
          { \mathsf{S}\downarrow}{B}{K_0 \vlbin{\vlne\alpha} K_1}} $
        exists such that, for every $ i\in\Set{0, 1} $,
        $ \vldownsmash
          {\vlder{\mathcal{Q}_i}
          { \mathsf{S}\downarrow}
          {A_i \vlbin{\beta} K_i}
          {\vlne{u_{\beta}}}}\enspace $ 
          exists as well
        and $ |\mathcal{Q}_0|+|\mathcal{Q}_1|\leq |\mathcal{P}| $.
    \end{enumerate}
\end{theorem}
\begin{proof}
We prove both points simultaneously, proceeding by induction on $ |\mathcal{P}| 
$. The value of $ |\mathcal{P}| $ is at least 1 because $ \alpha \prec \beta$
and $ \alpha $ is not weak. Necessarily, an occurrence of~\eqref{align:ax6 gen 
SA} exists in $ \mathcal{P} $ which generates a formula out of $ \vlne{u_\beta} 
$ with $ \alpha $ in it.

\begin{itemize}
    \item 
The base case is with $ |\mathcal{P}| = 1$ and~\eqref{align:ax6 gen SA} occurs 
in $ \mathcal{P} $. So, $ \mathcal{P} $ is composed by the three derivations
$ 
{\vlderivation{
\vliq{}{}
{(\vlne{u_\beta}\vlbin{\alpha}\vlne{u_\beta}) \vlbin{\beta} B'}
{\vlde{\mathcal{P}'}
{\mathsf{S}\downarrow}
{\vlne{u_\beta} \vlbin{\beta} B'}
{\vlhy{\vlne{u_{\beta}}}}}
}}$,\
$ \vlupsmash
{\vlderivation{
\vlde{\mathcal{P}''}{\mathsf{S}\downarrow}
{B}
{\vlhy{B'}}}
}$ and
$ \vlupsmash
{\vlderivation{
\vlde{\mathcal{P}_i}{\mathsf{S}\downarrow}
{A_i}
{\vlhy{\vlne{u_{\beta}}}}
}}$, for every $ i\in\{0, 1\} $,
where $ |\mathcal{P}'| = |\mathcal{P}''|= |\mathcal{P}_0|= |\mathcal{P}_1| = 
0$.
Lemma~\ref{lemma:Atomic Deduction for generalized subatomic systems} holds on 
$ \mathcal{P}'$. So, 
$ \vlupsmash
  { 
  {\vlder{\mathcal{D}'}{\mathsf{S}\downarrow}{B'}{u_\beta}}}$ 
implies the existence of $ \mathcal{D} $ which is
$ \vlupsmash
{\vlderivation{
\vlde{\mathcal{D}'}{\mathsf{S}\downarrow}
{ B'}
{\vliq{}{}
{u_\beta}
{\vlhy{u_{\beta} \vlbin{\vlne{\alpha}} u_{\beta}}}}
}}$.

Two cases are now possible.
\begin{itemize}
\item 
Let $ \alpha  $ be unitary.
For every $ i\in\Set{0, 1} $, the proof
$ \mathcal{Q}_i $ is
$\,\vlupsmash
{
{\vlderivation{
\vlde{\mathcal{P}_i}{\mathsf{S}\downarrow}
{A_i \vlbin{\beta} u_\beta}
{\vliq{}{}
{\vlne{u_\beta} \vlbin{\beta} u_\beta}
{\vlhy{ \vlne{u_{\beta}}}}}}}} \,$.
Moreover, 
$ |\mathcal{Q}_0| + |\mathcal{Q}_1| 
= |\mathcal{P}_0| + |\mathcal{P}_1| 
< |\mathcal{P}| = 1$ because none among 
$ \mathcal{Q}_0, \mathcal{Q}_1, \mathcal{P}_0 $ and 
$ \mathcal{P}_1 $ contain axioms that count $1 $.

   \item 
   If $ \alpha  $ is a right or a left weakening we proceed as here above, but 
   focusing only on one of the two proofs $ \mathcal{Q}_0 $ and 
   $ \mathcal{Q}_1 $.
 \end{itemize}     
 
\item 
 The inductive case has $ |\mathcal{P}| > 1 $.
 We only develop the details of the relevant cases.
 The first relevant case is a refinement of point (3) in the original proof of Shallow splitting 
 of \cite{Tubella2016,Tubella:2018:SPS:3176362.3173544}. The refinement 
 requires to consider the possibilities
 that we introduce a constant by distinguishing among unitary relations, right weakening and left weakening.
 \begin{itemize}
   \item 
   Let $ \alpha $ and $ \gamma $ be right weakening such that $ \mathcal{P} $ is
   $
     {\vlderivation{
       \vliq{\eqref{align:right weakening ax}}{}
       {(A_0\vlbin{\alpha} A_1)\vlbin{\beta} B_0\vlbin{\beta} B_1}
       {\vlde{\mathcal{P}'}{\mathsf{S}\downarrow}
       {(((A_0\vlbin{\alpha} A_1)\vlbin{\beta} B_0) 
           \vlbin{\gamma} C)
           \vlbin{\beta} B_1}
       {\vlhy{\vlne{u_\beta}}}}
    }} \,$.
   Because $|\mathcal{P}'| < |\mathcal{P}|$, by the inductive hypothesis
   $ \vlder{\mathcal{D}'}{\mathsf{S}\downarrow}
   {B_1}
   {K_l \vlbin{\vlne \gamma} K_r} $ 
   exists such that
   $ \vlderivation{
     \vlde
     {\mathcal{Q}'}{\mathsf{S}\downarrow}
     {(A_0\vlbin{\alpha} A_1)\vlbin{\beta} B_0\vlbin{\beta} K_l}
     {\vlhy{\vlne{u_\beta}}}
   }$ 
   exists as well with
   $ |\mathcal{Q}'| \leq |\mathcal{P}'| < |\mathcal{P}|  $.
   So, the inductive hypothesis holds on $\mathcal{Q}'$ and a derivation 
   $ \vlder{\mathcal{D}''}{\mathsf{S}\downarrow}
   {B_0\vlbin{\beta} K_l}
   {K_0 \vlbin{\vlne\alpha} K_1} $ exists such that
   $\vlupsmash
    {\vldownsmash
    {\vlderivation{
     \vlde
     {\mathcal{Q}''}{\mathsf{S}\downarrow}
     {A_0 \vlbin{\beta} K_0}
     {\vlhy{\vlne{u_\beta}}}}}}$ exists as well with 
   $ |\mathcal{Q}''| \leq |\mathcal{Q}'| \leq |\mathcal{P}'| < |\mathcal{P}|$.
   The proof we are looking for is $ \mathcal{Q}'' $.
   The derivation is
   $\,
   {\vlderivation{
    \vlde{\mathcal{D}'}{\mathsf{S}\downarrow}
    {B_0 \vlbin{\beta} B_1}
    {\vliq{\eqref{align:right weakening ax}}{}
    {B_0 \vlbin{\beta} (K_l \vlbin{\vlne\gamma} K_r)}
    {\vlde{\mathcal{D}''}{\mathsf{S}\downarrow}
    {B_0 \vlbin{\beta} K_l}
    {\vlhy{K_0 \vlbin{\vlne\alpha} K_1}}}}}}\, $.

   \item 
   Let $ \alpha $ be unitary and let $ \gamma $ be right weakening such that $ \mathcal{P} $ is
   $ \vlupsmash
     {\vlderivation{
      \vliq{\eqref{align:right weakening ax}}{}
      {(A_0\vlbin{\alpha} A_1)\vlbin{\beta} B_0\vlbin{\beta} B_1}
      {\vlpr{\mathcal{P}'}{\mathsf{S}\downarrow}
      {(((A_0\vlbin{\alpha} A_1)\vlbin{\beta} B_0) \vlbin{\gamma}  
      C)\vlbin{\beta} B_1}}
     }} \,$.
   Because $|\mathcal{P}'| < |\mathcal{P}|$, by the inductive hypothesis
   $ \vlder{\mathcal{D}'}{\mathsf{S}\downarrow}
   {B_1}
   {K_l \vlbin{\vlne \gamma} K_r} $ exists such that
   $ \vlderivation{
    \vlde{\mathcal{Q}'}{\mathsf{S}\downarrow}
    {(A_0\vlbin{\alpha} A_1)\vlbin{\beta} B_0\vlbin{\beta} K_l}
    {\vlhy{\vlne{u_\beta}}}
   } $ 
   exists as well with
   $ |\mathcal{Q}'| \leq |\mathcal{P}'| < |\mathcal{P}|  $.
   So, the inductive hypothesis holds on $\mathcal{Q}'$ and a derivation 
   $ \vlder{\mathcal{D}''}{\mathsf{S}\downarrow}
   {B_0\vlbin{\beta} K_l}
   {K_0 \vlbin{\vlne\alpha} K_1} $ exists such that, for every $ i\in\Set{0, 1} $,
   the proof 
   $\vlderivation{
    \vlde{\mathcal{Q}_i}{\mathsf{S}\downarrow}
    {A_i \vlbin{\beta} K_i}
    {\vlhy{\vlne{u_\beta}}}
   } 
   $ 
   exists as well and
   $ |\mathcal{Q}_0| +|\mathcal{Q}_1| \leq |\mathcal{Q}'| \leq |\mathcal{P}'| < |\mathcal{P}|$.
   The proofs we are looking for are $ \mathcal{Q}_0 $ and $ \mathcal{Q}_1 $.
   The derivation $ \mathcal{D} $ is
   $\, \vlderivation{
   \vlde{\mathcal{D}'}{\mathsf{S}\downarrow}
   {B_0 \vlbin{\beta} B_1}
   {\vliq{\eqref{align:right weakening ax}}{}
   {B_0 \vlbin{\beta} (K_l \vlbin{\vlne\gamma} K_r)}
   {\vlde{\mathcal{D}''}{\mathsf{S}\downarrow}
   {B_0 \vlbin{\beta} K_l}
   {\vlhy{K_0 \vlbin{\vlne\alpha} K_1}}}}}\, $.		  

   \item 
   The cases with both $ \alpha $ and $ \gamma $ left weakening or with
   $ \alpha $ unitary and $ \gamma $ left weakening are symmetric.
  \end{itemize}
  \end{itemize}
The further relevant cases come from points (13) and (14) in the original 
proof of Shallow splitting  of 
\cite{Tubella2016,Tubella:2018:SPS:3176362.3173544}. In our case, point 
(13) requires to focus also on a right weakening $ \alpha $ in a proof $ 
\mathcal{P} $ with form
$ \vlderivation{
  \vliq{\eqref{align:right weakening ax}}{}
  {(A_0\vlbin{\alpha} A_1)\vlbin{\beta} B}
  {\vlde{\mathcal{P}'}{\mathsf{S}\downarrow}
  {A_0\vlbin{\beta} B}
  {\vlhy{\vlne{u_\beta}}}
  }} \,$.
From $ \eqref{align:right weakening ax} $ we get that
$ \mathcal{D} $ is
$ \vlderivation{\vliq{\eqref{align:right weakening 
ax}}{}{B}{\vlhy{B\vlbin{\vlne\alpha}K}}} $.
So, the proof $ \mathcal{Q} $ is simply $ \mathcal{P}' $.
For the analogous of (14) with a left weakening it is enough to proceed as 
just done her above.
\end{proof}

\begin{definition}[Provable context]
    Let $ \beta $ be a weak relation with unit $ u_{\beta} $ in some Subatomic 
    system-1.1. A context $ H $ is \emph{provable} if $ 
    H\!\Set{\vlne{u_{\beta}}} = 
    \vlne{u_{\beta}}$.
    \qedh
\end{definition}
Theorem~\ref{theorem:Shallow Splitting for generalized subatomic systems}
implies that Context reduction holds exactly as formulated and proved in
\cite{Tubella2016,Tubella:2018:SPS:3176362.3173544}:

\begin{theorem}[Context reduction]
    \label{fact:Context reduction for S}
    Let $ \mathsf{S} $ be a Subatomic system-1.1 whose fragment 
    $ \mathsf{S}\!\!\downarrow $ is splittable.
    Let $ \beta $ a weak relation in $ \mathsf{S} $ with unit $ u_{\beta} $.
    For every $ A\in\mathcal{F} $ and context $ S $,
    if
    $ 
      {\vlderivation{
        \vlde{\mathcal{P}}{\mathsf{S}\downarrow}
        {S\{A\}}
        {\vlhy{\vlne {u_{\beta}}}}}} $, 
    then there is $ K\in\mathcal{F} $ and a provable context $ H$ such that 
    $ \vlupsmash
      {\vlderivation{
        \vlde{\mathcal{P}}{\mathsf{S}\downarrow}
        {S\{\, \}}
        {\vlhy{H\Set{\Set{\, }\vlbin{\beta} K}}}}} $ and
    $ \vlupsmash
       {\vlderivation{
        \vlde{\mathcal{P}}{\mathsf{S}\downarrow}
        {A \vlbin{\beta} B}
        {\vlhy{\vlne {u_{\beta}}}}}}$.
    \qedh
\end{theorem}

    
\begin{theorem}[Splittable up-fragment is admissible]
\label{theorem:Suparrow is admissible}
Let $ \mathsf{S} $ be a Subatomic system-1.1 with
splittable $ \mathsf{S}\!\!\downarrow $ and $ \mathsf{S}\!\!\uparrow $ in it.
Let $ A, B, C, D\in\mathcal{F} $ and $ S $ be a context.
Let $ \alpha \in\mathcal{R} $ be strong. 
For every $ \gamma\in\mathcal{R} $ such that $ \alpha\prec \gamma $,
if
$ \vlderivation{
    \vlin{\rho\uparrow}{}
    {S\{(A\vlbin{\alpha}C)\vlbin{\gamma}(B\vlbin{\alpha}D)\}}
    {\vlde{\mathcal{P}}{\mathsf{S}\downarrow}
    {S\{(A\vlbin{\gamma}B)\vlbin{\alpha}(C\vlbin{\vlne\gamma}D)\}}
    {\vlhy{\vlne{u_{\beta}}}}}
  } $\,, with $ \rho\!\!\uparrow $ in $ \mathsf{S}\!\!\uparrow $,
then
$ \vlupsmash{
   \vlderivation{
    \vlde{\mathcal{P}'}{\mathsf{S}\downarrow}
    {S\{(A\vlbin{\alpha}C)\vlbin{\gamma}(B\vlbin{\alpha}D)\}}
    {\vlhy{\vlne{u_{\beta}}}}
  }
} $ which means that $ \rho\!\!\uparrow $ is admissible in 
$ \mathsf{S}\!\!\downarrow $.
\end{theorem}
\begin{proof}
    We develop a case specific to version 1.1 where $ \gamma$ and, hence $ 
    \vlne{\gamma} $, is a right weakening.
%
    Theorem~\ref{fact:Context reduction for S} on $ \mathcal{P} $ implies
    $ \vlderivation{
        \vlde{\mathcal{D}}{\mathsf{S}\downarrow}
        {S\{\, \}}
        {\vlhy{H\Set{\Set{\, }\vlbin{\beta} K}}}} $ and
    $ \vlderivation{
        \vlde{\mathcal{Q}}{\mathsf{S}\downarrow}
        {((A\vlbin{\gamma}B)\vlbin{\alpha}(C\vlbin{\vlne\gamma}D)) \vlbin{\beta} K}
        {\vlhy{\vlne {u_{\beta}}}}}$ with $ H $ provable.
    Theorem~\ref{theorem:Shallow Splitting for generalized subatomic systems} on $ \mathcal{Q} $
    implies 
    $ \vlderivation{
        \vlde{\mathcal{D}'}{\mathsf{S}\downarrow}
        {K}
        {\vlhy{Q_1\vlbin{\beta}Q_2}}}$ and
    $ \vlderivation{
        \vlde{\mathcal{Q}_1}{\mathsf{S}\downarrow}
        {(A\vlbin{\gamma}B) \vlbin{\beta} Q_1}
        {\vlhy{\vlne {u_{\beta}}}}}$ and
    $ \vlderivation{
        \vlde{\mathcal{Q}_2}{\mathsf{S}\downarrow}
        {(C\vlbin{\vlne\gamma}D) \vlbin{\beta} Q_2}
        {\vlhy{\vlne {u_{\beta}}}}}$
    .
    Theorem~\ref{theorem:Shallow Splitting for generalized subatomic systems} on $ \mathcal{Q}_1 $
    implies
    $ \vlderivation{
        \vlde{\mathcal{D}_1}{\mathsf{S}\downarrow}
        {Q_1}
        {\vlhy{Q_A\vlbin{\vlne \gamma}Q_B}}}$ and
    $ \vlderivation{
        \vlde{\mathcal{Q}_A}{\mathsf{S}\downarrow}
        {A \vlbin{\beta} Q_A}
        {\vlhy{\vlne {u_{\beta}}}}}$
    .
    Theorem~\ref{theorem:Shallow Splitting for generalized subatomic systems} on $ \mathcal{Q}_2 $
    implies
    $ \vlderivation{
        \vlde{\mathcal{D}_2}{\mathsf{S}\downarrow}
        {Q_2}
        {\vlhy{Q_C\vlbin{\gamma}Q_D}}}$ and
    $ \vlderivation{
        \vlde{\mathcal{Q}_C}{\mathsf{S}\downarrow}
        {C \vlbin{\beta} Q_C}
        {\vlhy{\vlne {u_{\beta}}}}}$
    .
    So\
    $\vlderivation{
      \vlde{}{\mathsf{S}\downarrow}
      {S\!\Set{(A\vlbin{\gamma}B)\vlbin{\alpha}(C\vlbin{\vlne\gamma}D)}}
      {\vlde{}{\mathsf{S}\downarrow}
      {H\!\Set{((A\vlbin{\gamma}B)\vlbin{\alpha}(C\vlbin{\vlne\gamma}D))\vlbin{\beta}K}}
      {\vlde{}{\mathsf{S}\downarrow}
      {H\!\Set{((A\vlbin{\gamma}B)\vlbin{\alpha}(C\vlbin{\vlne\gamma}D))
              \vlbin{\beta}(Q_1\vlbin{\beta}Q_2)}}
      {\vlin{}{}
      {H\!\Set{((A\vlbin{\gamma}B)\vlbin{\alpha}(C\vlbin{\vlne\gamma}D))
                 \vlbin{\beta}((Q_A\vlbin{\vlne \gamma}Q_B)\vlbin{\beta}(Q_C\vlbin{\gamma}Q_D))}}
      {\vlin{}{}
      {H\!\Set{((A\vlbin{\gamma}B)\vlbin{\alpha}(C\vlbin{\vlne\gamma}D))
                 \vlbin{\beta}((Q_A\vlbin{\beta}Q_C)\vlbin{\vlne \gamma}(Q_B\vlbin{\beta}Q_D))}}
      {\vlin{}{}
      {H\!\Set{((A\vlbin{\alpha}C)\vlbin{\gamma}(B\vlbin{\alpha}D))
                \vlbin{\beta}((Q_A\vlbin{\beta}Q_C)\vlbin{\vlne \gamma}(Q_B\vlbin{\beta}Q_D))}}
      {\vliq{}{}
      {H\!\Set{((A\vlbin{\alpha}C)\vlbin{\beta}(Q_A\vlbin{\beta}Q_C))
               \vlbin{\gamma}((B\vlbin{\alpha}D)\vlbin{\beta}(Q_B\vlbin{\beta}Q_D))}}
      {\vlin{}{}
      {H\!\Set{(A\vlbin{\alpha}C)\vlbin{\beta}(Q_A\vlbin{\beta}Q_C)}}
      {\vlde{}{\mathsf{S}\downarrow}
      {H\!\Set{(A\vlbin{\beta}Q_A)\vlbin{\alpha}(C\vlbin{\beta}Q_C)}}
      {\vliq{}{}
      {H\!\Set{\vlne{u_{\beta}}\vlbin{\alpha}\vlne{u_{\beta}}}}
      {\vlhy{H\!\Set{\vlne{u_{\beta}}}}}}}}}}}}}}
      }\ .$
\end{proof}

\section{The splittable fragment $ \mathsf{P}\!\!\downarrow$ in $\mathsf{P}$}
\label{section:Splitting for Pdownarrow}
In this section we take advantage of having identified the properties that a 
Subatomic system-1.1 must meet to enjoy the cut-elimination property.
From Definitions~\ref{definition:splittable subatomic systems-1.1},
\ref{definition:Up-fragment of S} 
and~\ref{definition: Subatomic system 1.1 P} it follows:
\begin{fact}
    \label{fact:Properties of Sdownarrow}
    The Splittable \emph{down}-fragment $ \mathsf{P}\!\!\downarrow $ 
    of $ \mathsf{P}$ contains the \emph{down}-rules:
    \begin{center}
        \newcommand{\localGlb}   {\vlor}
        \newcommand{\localLub}   {\vlan}
        \newcommand{\localnotGlb}{\vlan}
        \newcommand{\localnotLub}{\vlor}
        \resizebox{.975\textwidth}{!}{
            \begin{tabular}{cccccc}
                $\vlinf{ai_0\downarrow}{} 
                {(A \vlbin{\pi_0} C)\vlbin{\localGlb}(B \vlbin{\pi_0} D)}
                {(A \vlbin{\localGlb} B)\vlbin{\pi_0}(C \vlbin{\localGlb} D)}$
                &\qquad\qquad&
                $\vlinf{ai_1\downarrow}{} 
                {(A \vlbin{\pi_1} C)\vlbin{\localGlb}(B \vlbin{\pi_1} D)}
                {(A \vlbin{\localGlb} B)\vlbin{\pi_1}(C \vlbin{\localGlb} D)}$
                &\qquad\qquad&
                $\vlinf{s\downarrow}{}
                {(A \vlbin{\localLub} C)\vlbin{\localGlb}(B \vlbin{\localnotLub} D)}
                {(A \vlbin{\localGlb} B)\vlbin{\localLub}(C \vlbin{\localGlb}    D)}$
            \end{tabular}
        } 
    \end{center}
    while the Splittable \emph{up}-fragment $ \mathsf{P}\!\!\uparrow $ 
    of $ \mathsf{P}$ contains the \emph{up}-rules:
    \begin{center}
    \newcommand{\localGlb}   {\vlan}
    \newcommand{\localLub}   {\vlor}
    \newcommand{\localnotGlb}{\vlor}
    \newcommand{\localnotLub}{\vlan}
    \resizebox{.975\textwidth}{!}{
        \begin{tabular}{cccccc}
            $\vlinf{ai_0\uparrow}{} 
            {(A \vlbin{\localGlb} C)\vlbin{\pi_0}(B \vlbin{\localGlb} D)}
            {(A \vlbin{\pi_0} B)\vlbin{\localGlb}(C \vlbin{\pi_0} D)}
            $
            &\qquad\qquad&
            $\vlinf{ai_1\uparrow}{} 
            {(A \vlbin{\localGlb} C)\vlbin{\pi_1}(B \vlbin{\localGlb} D)}
            {(A \vlbin{\pi_1} B)\vlbin{\localGlb}(C \vlbin{\pi_1} D)}
            $
            &\qquad\qquad&
            $\vlinf{s\uparrow}{}
            {(A \vlbin{\localGlb} C)\vlbin{\localLub}(B \vlbin{\localGlb}  D)}
            {(A \vlbin{\localLub} B)\vlbin{\localGlb}(C \vlbin{\localnotLub} D)}
            $
        \end{tabular}
    } 
\end{center}
\end{fact}
%
%
%
%
Theorem~\ref{theorem:Suparrow is admissible} holds on $ \mathsf{P} $, hence on 
the subset of rules of $ \mathsf{P}\!\!\downarrow $ and $ 
\mathsf{P}\!\!\uparrow $. So, we get:
\begin{corollary}
Every up-rule of \,$ \mathsf{P}\!\!\uparrow $ is admissible in 
$ \mathsf{P}\!\!\downarrow $.
\qedh
\end{corollary}
\vllinearfalse
\section{The system $ \mathsf{P} $ and Post's Lattice}
\label{section:Subatomic systems and Post's Lattice}
We show that $ \mathsf{P} $ is related to Post's Lattice \cite{Post1941}. 
It follows that $ \mathsf{P} $ extends Propositional logic without relying on 
any representations of the atoms of $ \mathsf{P} $ in terms of sub-atoms, i.e. 
in terms of some encoding which is based on self-dual non-commutative relations.

\begin{definition}[Clones \cite{Post1941}]
Let $ B $ be a set of boolean operators.
A \emph{clone} $ [B] $ is the least set of boolean operators of any arity,
closed under composition that contains:
(i) propositional variables $ x, y, z, \ldots $;
(ii) projections of every finite arity, $ \pi^1_1(x) = x $ included;
(iii) $ f \in B $ applied to propositional variables.
\qedh
\end{definition}
The class of all clones is Post's Lattice which is infinite and complete 
\cite{Post1941}. The top of the lattice is $ \mathsf{C_1} = [\vlor, \bTT, 
\vlne{\ }] $ which strictly contains five pairwise incomparable maximal clones:
\begin{align*}
&&\mathsf{C_2} & = [\vlan, \rightarrow, \leftarrow, \vlor, \bTT]\\
&&\mathsf{C_3} & = [\bFF, \vlan, \vlne{\leftarrow}, \vlne{\rightarrow}, \vlor]\\
&&\mathsf{L_1} & = [\bFF, \vlan, \vlor, \bTT ]\\
&&\mathsf{A_1} & = [\bFF, \Leftrightarrow, \oplus, \bTT, \vlne{\ }]\\
&&\mathsf{D_3} & = [\underbrace{(x\vlan y)\vlor(y\vlan z)\vlor(x\vlan 
z)}_{\operatorname{min}(x,y,z)}, 
         \underbrace{\vlne{(x\vlan y)\vlor(y\vlan z)\vlor(x\vlan z)}}_{\operatorname{min}(x,y,z)}]
\end{align*}
whose names come from \cite{Post1941,DBLP:journals/mst/Lewis79}.
%
%
%
%
%
%

\begin{proposition}[Soundness of $ \mathsf{P} $]
\label{proposition:Soundness of P wrt C1}
The Subatomic system-1.1 $ \mathsf{P} $ is sound for $ \mathsf{C_1} $.
\Ie, let $ A, B_1,\ldots,B_n \in \mathcal{F}_{\mathsf{P}}$ be such that $ B \in \mathsf{C_1}$ exists
and $ A $ is equivalent to $ B $, up to De Morgan equivalences.
Given  
$ \vlderivation{\vlder{\mathcal{P}}{\mathsf{P}}{A\vlor \vlne{B_1}\vlor \ldots\vlor \vlne{B_n}}{\bTT}} $,
if $ B_1\vlan \ldots\vlan B_n $ is true, then $ A $ is true.
\end{proposition}
\begin{proof}
    $ \mathcal{P} $ only contains rules of $ \mathsf{P} $.
    By definition, the conclusion of every rule in $ \mathsf{P} $ is true 
    whenever its premise is true. Since the formula on top of $ \mathcal{P} $ 
    is $ \bTT $ also 
    $ A\vlor \vlne{B_1}\vlor \ldots\vlor \vlne{B_n} =
      \vlne{B_1\vlan \ldots\vlan B_n} \vlor A = B_1\vlan \ldots\vlan B_n \vlbin{\Rightarrow} A $ 
    must be true. Forcefully, the truth of
     $ B_1\vlan \ldots\vlan B_n $ implies the one of $ A $.
\end{proof}
The proof of completeness follows a standard technique.
\begin{definition}
    Let $ A[x_1,\ldots,x_n] $ denote any formula of $ \mathcal{F}_{\mathsf{P} }$ such that 
    $ x_1,\ldots, x_n $ are all and only its variables. Let $ T_A $ be the following truth table
    of $ A[x_1,\ldots,x_n] $:
    \begin{center}
     \begin{tabular}{c|c|c|c||c}
         $ x_1 $  & $ x_2 $  & \ldots & $ x_n $  & $ A[x_1,\ldots,x_n] $ \\\hline\hline
         $ \bFF$  & $ \bFF$  & \ldots & $ \bFF$  & $ \chi_0 $\\\hline
         $ \bFF$  & $ \bFF$  & \ldots & $ \bTT$  & $ \chi_1 $\\\hline
         $\vdots$ & $\vdots$ &        & $\vdots$ & $\vdots  $\\\hline
         $ \bTT$  & $ \bTT$  & \ldots & $ \bTT$  & $ \chi_{2^n-1} $
     \end{tabular}
    \end{center}
    where $ \chi_l \in\Set{\bFF, \bTT} $, for every $0 \leq l \leq 2^n-1$.
    For every $ 0 \leq l \leq 2^n-1$ and every $ B \in \Set{x_1,\ldots,x, A[x_1,\ldots,x_n]} $
    let $ T_A(l, B) $ the entry of $ T_A $ at line $ l $ and column $ B $. By 
    definition, let $ \tau $ be the following map:
    \begin{align*}
       \tau(l, x_i)               & = x_i                     && \textrm{ if } T_A(l, 
       x_i) = \bTT\\
       \tau(l, x_i)               & = \vlne{x_i}              && \textrm{ if } T_A(l, 
       x_i) = \bFF\\
       \tau(l, A[x_1,\ldots,x_n]) & = A[x_1,\ldots,x_n]       && \textrm{ if } T_A(l, 
       A[x_1,\ldots,x_n]) = \bTT\\
       \tau(l, A[x_1,\ldots,x_n]) & = \vlne{A[x_1,\ldots,x_n]}&& \textrm{ if } T_A(l, 
       A[x_1,\ldots,x_n]) = \bFF
       \enspace .
    \end{align*}
\end{definition}

\begin{fact}[Arbitrary projections in $ \mathsf{P} $]
\label{fact:arbitrary projections}
For every $ x_i\in\mathcal{V}_{\mathsf{P}} $, the derivation
$ \vlder{}{}
  {\pi_i(\ldots,x_i,\ldots)}{x_i} $
exists by applying suitable combinations of the axioms~\eqref{align:axiom right 
weakening pi0} and~\eqref{align:axiom right weakening pi1}. The same holds for every 
$ \vlne{x_i}\in\vlne{\mathcal{V}}_{\mathsf{P}} $
\qedh
\end{fact}    

\begin{proposition}[Compactness of $ \mathsf{P} $]
\label{proposition:Compactness of P}
    Let $ A[x_1,\ldots,x_n] \in \mathcal{F}_{\mathsf{P}} $ be given. Then, for every 
    $ 0 \leq l \leq 2^n-1$, the proof
    $ \vlder{\mathcal{P}}
      {\mathsf{P}}
      {\tau(l, A[x_1,\ldots,x_n])\vlor\vlne{\tau(l, x_0)}\vlor\ldots\vlor\vlne{\tau(l, 
      x_n)}}
      {\bTT} $ exists.
\end{proposition}
\begin{proof}
We proceed by induction on the structure of $ A[x_1,\ldots,x_n] $.
\par
Let $ A[x_1,\ldots,x_n] $ be $ x_i $.
If $ \tau(l, x_i) = x_i $, then $ \mathcal{P} $ is
    $ \vlderivation{
        \vliq{}{}
        {\tau(l, x_i)\vlor\vlne{\tau(l, x_i)}}
        {\vliq{\eqref{align:axiom excluded middle p}}{}
            {x_i\vlor\vlne{x_i}}
            {\vlhy{\bTT}}}
    } $.
If $ \tau(l, x_i) = \vlne{x_i} $, then $ \mathcal{P} $ is
    $ \vlupsmash
      {
      {\vlderivation{
      \vliq{}{}
      {\vlne{\tau(l, x_i)}\vlor\tau(l, x_i)}
      {\vliq{\eqref{align:axiom excluded middle p}}{}
      {\vlne{x_i}\vlor x_i}
      {\vlhy{\bTT}}}
    }}} $.
Let $ A[x_1,\ldots,x_n] $ be $ \pi_i(x_1,\ldots,x_n) $. 
If $ \tau(l, \pi_i(x_1,\ldots,x_n)) = \pi_i(x_1,\ldots,x_n) $, then 
$ \mathcal{P} $ is
    $$\vlupsmash 
      {\vlderivation{
      \vliq{\textrm{Fact}~\ref{fact:arbitrary projections}}{}
      {\tau(l, \pi_i(x_1,\ldots,x_n))
            \vlor\vlne{\tau(l, \pi_i(x_1,\ldots,x_n))}}
      {\vliq{\eqref{align:axiom de morgan pi}}{}
      {\pi_i(\ldots,x_i,\ldots)\vlor\vlne{\pi_i(\ldots,x_i,\ldots)}}
      {\vliq{}{}
      {\pi_i(\ldots,x_i,\ldots)\vlor\pi_i(\ldots,\vlne{x_i},\ldots)}
      {\vliq{}{}
      {x_i\vlor\vlne{x_i}}
      {\vlhy{\bTT}}}}}
    } }
\enspace .$$
If $ A\in\Set{A_0\vlan A_1,A_0\vlor A_1 } $ it is enough to standardly apply 
the inductive hypothesis.
\end{proof}

\begin{theorem}[Completeness of $ \mathsf{P} $]
\label{proposition:Completeness of P wrt C1}
The Subatomic system-1.1 $ \mathsf{P} $ is complete for $ \mathsf{C_1} $.
\Ie, let $ A, B_1,\ldots,B_n \in \mathcal{F}_{\mathsf{P}}$ be such that $ B \in \mathsf{C_1}$ exists
and $ A $ is equivalent to $ B $, up to De Morgan equivalences.
Let us also assume that the truth of $ B_1\vlan \ldots\vlan B_n $ implies the truth of $ A $.
Then $ \vlderivation{\vlder{\mathcal{D}}{\mathsf{P}}{A\vlor \vlne{B_1}\vlor\ldots\vlor\vlne{B_n}}{\bTT}} $ exists.
\end{theorem}
\begin{proof}
The assumption saying that the truth of $ B_1\vlan \ldots\vlan B_n $ implies
the truth of $ A $ is equivalent to saying that $ A\vlor \vlne{B_1}\vlor \ldots\vlor \vlne{B_n} $
is a tautology.
To keep the proof readable we assume that $ x, y $ are all and only the free variables of 
$ A\vlor \vlne{B_1}\vlor \ldots\vlor \vlne{B_n} $. Of course, what we are going to do, works
for any finite set of variables in $ A\vlor \vlne{B_1}\vlor \ldots\vlor \vlne{B_n} $.
Proposition~\ref{proposition:Compactness of P} assures the existence of
$ 
  {\vlderivation{\vlder{\mathcal{P}_l}{\mathsf{P}}
  {\tau(l,X)
   \vlor\tau(l,x)\vlor \tau(l,y)}
  {\bTT}}}$ for every $ 1\leq l\leq 2^2 $, where $ X $ shortens
$ (A\vlor \vlne{B_1}\vlor \ldots\vlor \vlne{B_n})[x,y] $ and $ 2^2 $ is the 
number of lines that all the 
combinations of the literals $ x, \vlne{x}, y $ and $ \vlne{y} $ generate in the truth
table of $ \tau(l,X) $.
%
In fact, for every $ 1\leq l\leq 4 $, the proof $ \mathcal{P}_l $ has form
$ \vlderivation{\vlder{\mathcal{P}_l}{\mathsf{P}}
    {X\vlor\tau(l,x)\vlor\tau(l,y)}
    {\bTT}} $ because $ X$, \ie $A\vlor \vlne{B_1}\vlor \ldots\vlor \vlne{B_n} $,
is a tautology. So $ \mathcal{D} $ we are looking for is
$
   \vlderivation{
    \vlin{}{}
    {X}
    {\vliq{}{}
    {X\vlor X}
    {\vliq{}{}
    {\bFF\vlor X\vlor X}
    {\vlin{s\downarrow}{}
    {(y \vlan \vlne{y})\vlor X\vlor X}
    {\vliq{}{}
    {(X \vlor y)\vlan(X\vlor\vlne{y})}
    {\vlin{}{}
    {(X\vlor X\vlor y\vlor y)\vlan(X\vlor X\vlor \vlne{y}\vlor \vlne{y})}
    {\vliq{}{}
    {(\bFF\vlor(X\vlor X\vlor y\vlor y))\vlan(\bFF\vlor(X\vlor X\vlor \vlne{y}\vlor \vlne{y}))}
    {\vlin{s\downarrow,s\downarrow}{}
    {((x\vlan \vlne{x})\vlor(X\vlor X\vlor y\vlor y))\vlan((x\vlan \vlne{x})\vlor(X\vlor X\vlor \vlne{y}\vlor \vlne{y}))}
    {\vlde{}{}
    {((X\vlor x\vlor y)\vlan(X\vlor \vlne{x}\vlor y))\vlan((X\vlor x\vlor \vlne{y})\vlan(X\vlor \vlne{x}\vlor \vlne{y}))}
    {\vliq{}{}
    {\bTT\vlan\bTT\vlan\bTT\vlan\bTT}
    {\vlhy{\bTT}}}}}}}}}}}
}\enspace . $
\end{proof}

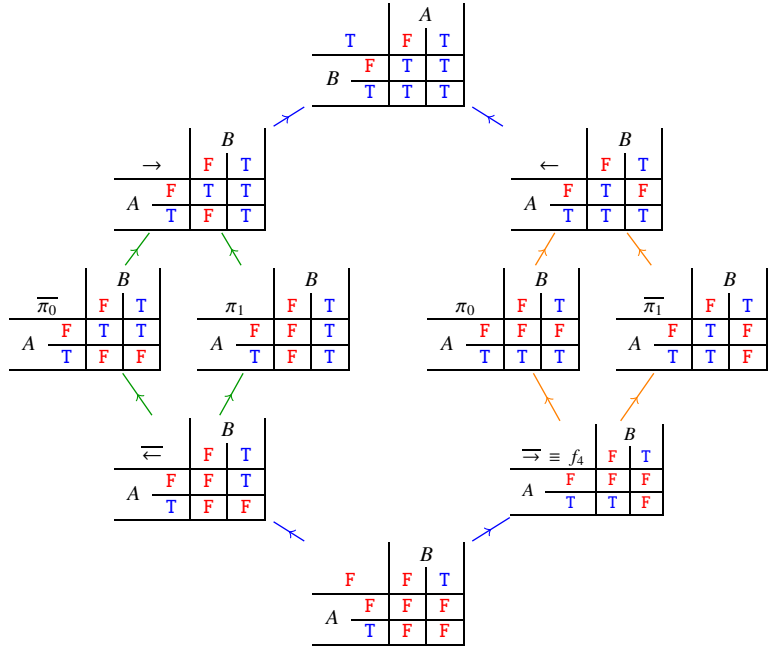
\vllinearfalse
\section{Conclusion and developments}
\label{section:conclusion}
This work highlights how much effective the work in 
\cite{Tubella2016,Tubella:2018:SPS:3176362.3173544} 
that aims at identifying the core 
mechanism of cut-elimination is. The notion of Subatomic system allows to prove 
modularly, generally and once the cut-elimination of interesting deep inference 
systems. We show that the original notion of subatomic system can be slightly 
generalized. This allows to identify new logical systems without any need to 
encode their constants sub-atomically and without loosing splitting, \ie 
cut-elimination. 
The Subatomic system-1.1 $ \mathsf{P} $, sound and complete for the tautologies 
of Post's clone $ \mathsf{C_1} $, is a witness of how that is possible. Of 
course, the introduction of $ \mathsf{P} $ is not breathtaking, 
but logical systems that smoothly incorporate self-dual operators --- in the 
case of $ \mathsf{P} $ they are operators as natural as projections --- and 
which keep maintaining good logical properties are not so common 
\cite{Gugl:06:A-System:kl,Roversi:TLCA11,Roversi01042016}.

On going work aims at using the framework of subatomic systems, may be upgraded 
to some version \texttt{x.y} --- this work introduces release 1.1 ---, for 
systematically identifying logical systems with good properties and, possibly, 
of some relevance. Saying it in another way, the idea is to use the pattern 
that subatomic systems suggest for contributing to systematic proof theory 
\cite{DBLP:conf/csl/CiabattoniST09}.
The following list of medial shapes should show to what extent this idea can be 
concrete and potentially interesting:
\begin{center}
    \hspace{-1.25cm}
 \begin{tabular}{cc}
        \newcommand{\localF}     {\vlne{\pi_0}}
        \newcommand{\localFp}    {\pi_1}
        \newcommand{\localGlb}   {\rightarrow}
        \newcommand{\localLub}   {\vlne{\leftarrow}}
        \newcommand{\localnotF}  {\vlne{\pi_0}}
        \newcommand{\localnotFp} {\pi_1}
        \newcommand{\localnotGlb}{\vlne{\rightarrow}}
        \newcommand{\localnotLub}{\leftarrow}
        \begin{tabular}{cc}
            $\vlinf{ai_0\downarrow}{} 
            {(A \vlbin{\localF} C)\vlbin{\localGlb}(B \vlbin{\localnotF} D)}
            {(A \vlbin{\localGlb} B)\vlbin{\localF}(C \vlbin{\localGlb} D)}$
            &
            $\vlinf{ai_0\uparrow}{} 
            {(A \vlbin{\localnotGlb} B)\vlbin{\localnotF}(C 
            \vlbin{\localnotGlb} D)}
            {(A \vlbin{\localnotF} C)\vlbin{\localnotGlb}(B \vlbin{\localF} D)}$
            \\\\
            $\vlinf{ai_1\downarrow}{} 
            {(A \vlbin{\localFp} C)\vlbin{\localGlb}(B \vlbin{\localnotFp} D)}
            {(A \vlbin{\localGlb} B)\vlbin{\localFp}(C \vlbin{\localGlb} D)}$
            &
            $\vlinf{ai_1\uparrow}{} 
            {(A \vlbin{\localnotGlb} B)\vlbin{\localnotFp}(C 
            \vlbin{\localnotGlb} D)}
            {(A \vlbin{\localnotFp} C)\vlbin{\localnotGlb}(B \vlbin{\localFp} 
            D)}$
            \\\\
            $\vlinf{s\downarrow}{}
            {(A \vlbin{\localLub} C)\vlbin{\localGlb}(B \vlbin{\localnotLub} 
D)}
            {(A \vlbin{\localGlb} B)\vlbin{\localLub}(C \vlbin{\localGlb}    
            D)}$
            &
            $\vlinf{s\uparrow}{}
            {(A \vlbin{\localnotGlb} B)\vlbin{\localnotLub}(C 
            \vlbin{\localnotGlb} D)}
            {(A \vlbin{\localnotLub} C)\vlbin{\localnotGlb}(B 
            \vlbin{\localLub}    D)}$
            \\\\
            $\vlinf{c_0\downarrow}{}
            {(A \vlbin{\localGlb} C)\vlbin{\localF  }(B \vlbin{\localGlb} D)} 
            {(A \vlbin{\localF  } B)\vlbin{\localGlb}(C \vlbin{\localF  } D)}$
            &
            $\vlinf{c_0\uparrow}{}
            {(A \vlbin{\localnotF  } B)\vlbin{\localnotGlb}(C 
            \vlbin{\localnotF  } D)} 
            {(A \vlbin{\localnotGlb} C)\vlbin{\localnotF  }(B 
            \vlbin{\localnotGlb} D)}$
        \end{tabular}
&   \hspace{-1cm}
        \newcommand{\localF}     {\pi_0}
        \newcommand{\localFp}    {\vlne{\pi_1}}
        \newcommand{\localGlb}   {\leftarrow}
        \newcommand{\localLub}   {\vlne{\rightarrow}}
        \newcommand{\localnotF}  {\pi_0}
        \newcommand{\localnotFp} {\vlne{\pi_1}}
        \newcommand{\localnotGlb}{\vlne{\rightarrow}}
        \newcommand{\localnotLub}{\rightarrow}
        %
        \begin{tabular}{cc}
            $\vlinf{ai_0\downarrow}{} 
            {(A \vlbin{\localF} C)\vlbin{\localGlb}(B \vlbin{\localnotF} D)}
            {(A \vlbin{\localGlb} B)\vlbin{\localF}(C \vlbin{\localGlb} D)}$
            &
            $\vlinf{ai_0\uparrow}{} 
            {(A \vlbin{\localnotGlb} B)\vlbin{\localnotF}(C 
            \vlbin{\localnotGlb} D)}
            {(A \vlbin{\localnotF} C)\vlbin{\localnotGlb}(B \vlbin{\localF} D)}$
            \\\\
            $\vlinf{ai_1\downarrow}{} 
            {(A \vlbin{\localFp} C)\vlbin{\localGlb}(B \vlbin{\localnotFp} D)}
            {(A \vlbin{\localGlb} B)\vlbin{\localFp}(C \vlbin{\localGlb} D)}$
            &
            $\vlinf{ai_1\uparrow}{} 
            {(A \vlbin{\localnotGlb} B)\vlbin{\localnotFp}(C 
            \vlbin{\localnotGlb} D)}
            {(A \vlbin{\localnotFp} C)\vlbin{\localnotGlb}(B \vlbin{\localFp} 
            D)}$
            \\\\
            $\vlinf{s\downarrow}{}
            {(A \vlbin{\localLub} C)\vlbin{\localGlb}(B \vlbin{\localnotLub} 
D)}
            {(A \vlbin{\localGlb} B)\vlbin{\localLub}(C \vlbin{\localGlb}    
            D)}$
            &
            $\vlinf{s\uparrow}{}
            {(A \vlbin{\localnotGlb} B)\vlbin{\localnotLub}(C 
            \vlbin{\localnotGlb} D)}
            {(A \vlbin{\localnotLub} C)\vlbin{\localnotGlb}(B 
            \vlbin{\localLub}    D)}$
            \\\\
            $\vlinf{c_0\downarrow}{}
            {(A \vlbin{\localGlb} C)\vlbin{\localF  }(B \vlbin{\localGlb} D)} 
            {(A \vlbin{\localF  } B)\vlbin{\localGlb}(C \vlbin{\localF  } D)}$
            &
            $\vlinf{c_0\uparrow}{}
            {(A \vlbin{\localnotF  } B)\vlbin{\localnotGlb}(C 
            \vlbin{\localnotF  } D)} 
            {(A \vlbin{\localnotGlb} C)\vlbin{\localnotF  }(B 
            \vlbin{\localnotGlb} D)}$
        \end{tabular}
\end{tabular}
\end{center}

\begin{center}
    \hspace{-1.25cm}
    \begin{tabular}{cc}
\newcommand{\localF}     {\vlne{\pi_0}}
\newcommand{\localFp}    {\pi_1}
\newcommand{\localGlb}   {\rightarrow}
\newcommand{\localLub}   {\vlne{\leftarrow}}
\newcommand{\localnotF}  {\vlne{\pi_0}}
\newcommand{\localnotFp} {\pi_1}
\newcommand{\localnotGlb}{\vlne{\rightarrow}}
\newcommand{\localnotLub}{\leftarrow}
\begin{tabular}{cc}
    $\vlinf{c_1\downarrow}{}
    {(A \vlbin{\localGlb} C)\vlbin{\localFp }(B \vlbin{\localGlb} D)} 
    {(A \vlbin{\localFp } B)\vlbin{\localGlb}(C \vlbin{\localFp } D)}$
    &
    $\vlinf{c_1\uparrow}{}
    {(A \vlbin{\localnotFp } B)\vlbin{\localnotGlb}(C 
        \vlbin{\localnotFp } D)} 
    {(A \vlbin{\localnotGlb} C)\vlbin{\localnotFp }(B 
        \vlbin{\localnotGlb} D)}$
\end{tabular}
&   \hspace{-1cm}
\newcommand{\localF}     {\pi_0}
\newcommand{\localFp}    {\vlne{\pi_1}}
\newcommand{\localGlb}   {\leftarrow}
\newcommand{\localLub}   {\vlne{\rightarrow}}
\newcommand{\localnotF}  {\pi_0}
\newcommand{\localnotFp} {\vlne{\pi_1}}
\newcommand{\localnotGlb}{\vlne{\rightarrow}}
\newcommand{\localnotLub}{\rightarrow}
%
\begin{tabular}{cc}
    $\vlinf{c_1\downarrow}{}
    {(A \vlbin{\localGlb} C)\vlbin{\localFp }(B \vlbin{\localGlb} D)} 
    {(A \vlbin{\localFp } B)\vlbin{\localGlb}(C \vlbin{\localFp } D)}$
    &
    $\vlinf{c_1\uparrow}{}
    {(A \vlbin{\localnotFp } B)\vlbin{\localnotGlb}(C 
        \vlbin{\localnotFp } D)} 
    {(A \vlbin{\localnotGlb} C)\vlbin{\localnotFp }(B 
        \vlbin{\localnotGlb} D)}$ 
\end{tabular}
\end{tabular}
\end{center}
\noindent
The whole list is candidate to become a subatomic system 
$ \mathsf{R_{23}} $. Endowed it with the right equational theory among the 
propositional logic formulas that the rules infer, $ \mathsf{R_{23}} $ should 
derive tautologies of $ \mathsf{C_2} \cup \mathsf{C_3} $ we recall in 
Section~\ref{section:Subatomic systems and Post's Lattice}.
Lewis shows that the satisfiability of formulas
of $ \mathsf{C_2} $ belongs to P-Time problems while the satisfiability of
formulas in $ \mathsf{C_3} $ is NP-Time 
complete~\cite{DBLP:journals/mst/Lewis79}. So, $ \mathsf{R_{23}} $ would 
be a logical system where looking for proof-theoretical properties that highlight the 
phase transition from P-Time to NP-Time complete satisfiability.

The rules of $ \mathsf{R_{23}} $ come from the following complete 
sub-lattice:
\begin{center}
    \begin{tikzpicture}[scale=.55]
    \begin{pgfonlayer}{nodelayer}
    \node [style=none] (imp)    at (-4.750,-3.0) {
        \resizebox{.125\textwidth}{!}{
            \begin{tabular}{cc|c|c|c|l}
            & & \multicolumn{2}{ c| }{$ B $}        \\ 
            \multicolumn{2}{c|}{$ \vlne{\leftarrow} $} 
            & \bFF & \bTT \\ \hline
            \multicolumn{1}{ c  }{\multirow{2}{*}{$ A$} } 
            & \multicolumn{1}{ c| }{\bFF} & \bFF & \bTT \\ \cline{2-4}
            \multicolumn{1}{ c  }{} 
            & \multicolumn{1}{ c| }{\bTT} & \bFF & \bFF \\ \cline{1-4}
            \end{tabular}	
    }};
    \node [style=none] (coimp)  at ( 4.750,-3.0) {
        \resizebox{.125\textwidth}{!}{
            \begin{tabular}{cc|c|c|c|l}
            & & \multicolumn{2}{ c| }{$ B $}        \\ 
            \multicolumn{2}{c|}{$ \vlne{\rightarrow} \,\equiv\,f_4 $} 
            & \bFF & \bTT \\ \hline
            \multicolumn{1}{ c  }{\multirow{2}{*}{$ A $} } 
            & \multicolumn{1}{ c| }{\bFF} & \bFF & \bFF \\ \cline{2-4}
            \multicolumn{1}{ c  }{} 
            & \multicolumn{1}{ c| }{\bTT} & \bTT & \bFF \\ \cline{1-4}
            \end{tabular}	
    }};
    
    \node [style=none] (pi0n)   at (-7.25, 0.615) {
        \resizebox{.125\textwidth}{!}{
            \begin{tabular}{cc|c|c|c|l}
            & & \multicolumn{2}{ c| }{$ B$}        \\ 
            \multicolumn{2}{c|}{$ \vlne{\pi_0}$} 
            & \bFF & \bTT \\ \hline
            \multicolumn{1}{ c  }{\multirow{2}{*}{$ A $} } 
            & \multicolumn{1}{ c| }{\bFF} & \bTT & \bTT \\ \cline{2-4}
            \multicolumn{1}{ c  }{} 
            & \multicolumn{1}{ c| }{\bTT} & \bFF & \bFF \\ \cline{1-4}
            \end{tabular}	
    }};
    \node [style=none] (pi1)    at (-2.75, 0.615) {
        \resizebox{.125\textwidth}{!}{
            \begin{tabular}{cc|c|c|c|l}
            & & \multicolumn{2}{ c| }{$ B $}        \\ 
            \multicolumn{2}{c|}{$ \pi_1 $} 
            & \bFF & \bTT \\ \hline
            \multicolumn{1}{ c  }{\multirow{2}{*}{$ A $} } 
            & \multicolumn{1}{ c| }{\bFF} & \bFF & \bTT \\ \cline{2-4}
            \multicolumn{1}{ c  }{} 
            & \multicolumn{1}{ c| }{\bTT} & \bFF & \bTT \\ \cline{1-4}
            \end{tabular}	
    }};
    
    \node [style=none] (ff)     at ( 0.0 ,-6.0) {
        \resizebox{.125\textwidth}{!}{
            \begin{tabular}{cc|c|c|c|l}
            & & \multicolumn{2}{ c| }{$ B $}        \\ 
            \multicolumn{2}{c|}{$ \bFF $} 
            & \bFF & \bTT \\ \hline
            \multicolumn{1}{ c  }{\multirow{2}{*}{$ A $} } 
            & \multicolumn{1}{ c| }{\bFF} & \bFF & \bFF \\ \cline{2-4}
            \multicolumn{1}{ c  }{} 
            & \multicolumn{1}{ c| }{\bTT} & \bFF & \bFF \\ \cline{1-4}
            \end{tabular}	
    }};
    \node [style=none] (tt)    at ( 0.0 , 7.0) {
        \resizebox{.125\textwidth}{!}{
            \begin{tabular}{cc|c|c|c|l}
            & & \multicolumn{2}{ c| }{$ A $}        \\ 
            \multicolumn{2}{c|}{$ \bTT$} 
            & \bFF & \bTT \\ \hline
            \multicolumn{1}{ c  }{\multirow{2}{*}{$ B $} } 
            & \multicolumn{1}{ c| }{\bFF} & \bTT & \bTT \\ \cline{2-4}
            \multicolumn{1}{ c  }{} 
            & \multicolumn{1}{ c| }{\bTT} & \bTT & \bTT \\ \cline{1-4}
            \end{tabular}	
    }};
    
    \node [style=none] (pi0)    at ( 2.75, 0.615) {
        \resizebox{.125\textwidth}{!}{
            \begin{tabular}{cc|c|c|c|l}
            & & \multicolumn{2}{ c| }{$ B $}        \\ 
            \multicolumn{2}{c|}{$ \pi_0$} 
            & \bFF & \bTT \\ \hline
            \multicolumn{1}{ c  }{\multirow{2}{*}{$ A $} } 
            & \multicolumn{1}{ c| }{\bFF} & \bFF & \bFF \\ \cline{2-4}
            \multicolumn{1}{ c  }{} 
            & \multicolumn{1}{ c| }{\bTT} & \bTT & \bTT \\ \cline{1-4}
            \end{tabular}	
    }};
    \node [style=none] (pi1n)   at ( 7.25, 0.615) {
        \resizebox{.125\textwidth}{!}{
            \begin{tabular}{cc|c|c|c|l}
            & & \multicolumn{2}{ c| }{$ B $}        \\ 
            \multicolumn{2}{c|}{$ \vlne{\pi_1} $} 
            & \bFF & \bTT \\ \hline
            \multicolumn{1}{ c  }{\multirow{2}{*}{$ A $} } 
            & \multicolumn{1}{ c| }{\bFF} & \bTT & \bFF \\ \cline{2-4}
            \multicolumn{1}{ c  }{} 
            & \multicolumn{1}{ c| }{\bTT} & \bTT & \bFF \\ \cline{1-4}
            \end{tabular}	
    }};
    
    \node [style=none] (conimp)   at (-4.750, 4.0) {
        \resizebox{.125\textwidth}{!}{
            \begin{tabular}{cc|c|c|c|l}
            & & \multicolumn{2}{ c| }{$ B $}        \\ 
            \multicolumn{2}{c|}{$ \rightarrow$} 
            & \bFF & \bTT \\ \hline
            \multicolumn{1}{ c  }{\multirow{2}{*}{$ A $} } 
            & \multicolumn{1}{ c| }{\bFF} & \bTT & \bTT \\ \cline{2-4}
            \multicolumn{1}{ c  }{} 
            & \multicolumn{1}{ c| }{\bTT} & \bFF & \bTT \\ \cline{1-4}
            \end{tabular}	
    }};
    \node [style=none] (nimp) at ( 4.750, 4.0) {
        \resizebox{.125\textwidth}{!}{
            \begin{tabular}{cc|c|c|c|l}
            & & \multicolumn{2}{ c| }{$ B $}        \\ 
            \multicolumn{2}{c|}{$ \leftarrow$} 
            & \bFF & \bTT \\ \hline
            \multicolumn{1}{ c  }{\multirow{2}{*}{$ A $} } 
            & \multicolumn{1}{ c| }{\bFF} & \bTT & \bFF \\ \cline{2-4}
            \multicolumn{1}{ c  }{} 
            & \multicolumn{1}{ c| }{\bTT} & \bTT & \bTT \\ \cline{1-4}
            \end{tabular}	
    }};
    \end{pgfonlayer}
    \begin{pgfonlayer}{edgelayer}
    \draw [style=arrow, color=mygreen]  (imp) to (pi0n);
    \draw [style=arrow, color=mygreen]  (imp) to (pi1);
    \draw [style=arrow, color=mygreen]  (pi0n) to (conimp);
    \draw [style=arrow, color=mygreen]  (pi1) to (conimp);
    
    
    \draw [style=arrow, color=orange]  (coimp) to (pi0);
    \draw [style=arrow, color=orange]  (coimp) to (pi1n);
    \draw [style=arrow, color=orange]  (pi0) to (nimp);
    \draw [style=arrow, color=orange]  (pi1n) to (nimp);
    
    \draw [style=arrow, color=blue]  (conimp) to (tt);
    \draw [style=arrow, color=blue]  (nimp) to (tt);
    \draw [style=arrow, color=blue]  (ff) to (imp);
    \draw [style=arrow, color=blue]  (ff) to (coimp);
    \end{pgfonlayer}
    \end{tikzpicture}
\end{center}
which is inside the complete lattice of binary boolean functions pointwise 
ordered in accordance with the convention that $ \bFF $ is smaller than $ \bTT $.
The lattice shows that it is natural to work with more than one weak relation in the 
same system.
Both $ \vlne{\leftarrow} $ and $ \vlne{\rightarrow} $ are weak and play the same role as 
that played by $ \vlan $ in the lattice that drives the definition of 
$ \mathsf{P} $ in Section~\ref{section:The system P}. Two weak relations are required 
because the negation $ \leftarrow  $ of $ \vlne{\leftarrow} $ is the least upper 
bound of $ \pi_0 $ and $ \vlne{\pi_1} $ and not of $ \vlne{\pi_0} $ and $ \pi_1 $
of which $ \vlne{\leftarrow} $ is greatest lower bound. Of course, symmetrically, the 
same observation holds for $ \vlne{\rightarrow} $.
 
The lattice here above should immediately suggests that the search of subatomic systems 
need not be confined to the set of sixteen two-valued boolean functions. For any $ 
k\geq 3 $, the use of $ k $-valued operators as relations for subatomic systems is 
perfectly viable. For example, the subatomic system that corresponds to the  
paradigmatic deep inference system $ \mathsf{BV} $ \cite{Gugl:06:A-System:kl} can be 
seen as a system that uses 3-valued operators that define Coherence Spaces 
\cite{Girard:1989:PT:64805}. Considered the huge number of $ k $-valued operators, as $ 
k $ grows, subatomic systems look like grammars that generate specific languages, \ie 
logical systems, with good proof theoretical properties, of possible unexpected 
interest, as consequence of the consistent use of non standard logical operators.
This should definitely make it evident the contribution that the introduction of 
Subatomic systems-1.1 can give to Systematic Proof Theory.



\bibliography{bibliography}

\end{document}